\documentclass[twocolumn,pra,aps]{revtex4}

\usepackage{cleveref}
\usepackage{graphicx}
\usepackage{bm}
\usepackage{amsmath}
\usepackage{amssymb}
\usepackage{underscore}
\usepackage{amsthm}
\usepackage{hyperref}
\hypersetup{colorlinks, citecolor=magenta, filecolor=blue, linkcolor=blue, urlcolor=green}

\def\squareforqed{\hbox{\rlap{$\sqcap$}$\sqcup$}}
\def\qed{\ifmmode\squareforqed\else{\unskip\nobreak\hfil
\penalty50\hskip1em\null\nobreak\hfil\squareforqed
\parfillskip=0pt\finalhyphendemerits=0\endgraf}\fi}

\def\duzomniejsze{<\kern-.7mm<}
\def\duzowieksze{>\kern-.7mm>}

\def\textbf#1{{\bf #1}}
\def\beq{\begin{equation}}
\def\eeq{\end{equation}}
\def\be{\begin{equation}}
\def\ee{\end{equation}}
\def\bal{\begin{align}}
\def\eal{\end{align}}
\def\ben{\begin{eqnarray}}
\def\een{\end{eqnarray}}
\def\beqa{\begin{eqnarray}}
\def\eeqa{\end{eqnarray}}
\def\eea{\end{array}}
\def\bea{\begin{array}}

\newcommand{\bei}{\begin{itemize}}
\newcommand{\eei}{\end{itemize}}
\newcommand{\bee}{\begin{enumerate}}
\newcommand{\eee}{\end{enumerate}}
\newcommand{\nc}{\newcommand}
\nc{\1}{{\openone}}

\def\zcal{{\cal Z}}

\def\hcal{{\cal H}}

\def\CC{{\cal C}}

\def\tr{{\rm Tr}}

\def\>{\rangle}
\def\<{\langle}

\def\ot{\otimes}

\newtheorem{lemma}{Lemma}
\newtheorem{corollary}{Corollary}
\newtheorem{proposition}{Proposition}
\newtheorem{theorem}{Theorem}
\newtheorem{definition}{Definition}
\newtheorem{observation}{Observation}

\newtheorem{rem}{Remark}
\newtheorem{example}{Example}

\theoremstyle{plain}

\newtheorem*{Th}{Theorem 1}

\def\bed{\begin{definition}}
\def\eed{\end{definition}}
\def\bel{\begin{lemma}}
\def\eel{\end{lemma}}
\def\bet{\begin{theorem}}
\def\eet{\end{theorem}}

\usepackage[usenames, dvipsnames]{color}

\begin{document}

\title{ On distilling secure key from reducible private states and (non)existence of entangled key-undistillable states}

\begin{abstract}
Hereunder, we study the class of irreducible private states that are private states from which all the secret content is accessible via
measuring their key part. We provide the first protocol which distills key not only from the key part, but also from the shield if only
the state is reducible. We prove also a tighter upper bound on the performance of that protocol, given in terms of regularized relative entropy of entanglement instead of relative entropy of entanglement previously known.
This implies in particular that the irreducible private states are all strictly irreducible if and only if the entangled but key-undistillable states ("entangled key-undistilable states") exist. In turn, all the irreducible private states of the dimension $4\otimes 4$ are strictly irreducible, that is, after an attack on the key part they become separable. Provided the bound key states exist, we consider different subclasses of the irreducible private states and their properties. Finally we provide a lower bound on the trace norm distance between key-undistillable states and private states, in sufficiently high dimensions.
\end{abstract}

\author{K. Horodecki}
\affiliation{Institute of Informatics, National Quantum Information Centre, Faculty of
Mathematics, Physics and Informatics, University of Gda{\'n}sk, 80-308 Gda{\'n}sk, Poland}
\author{P. \'Cwikli\'nski}
\affiliation{Institute of Theoretical Physics and Astrophysics, National Quantum Information Centre, Faculty of
Mathematics, Physics and Informatics, University of Gda{\'n}sk, 80-308 Gda{\'n}sk, Poland}
\author{A. Rutkowski}
\affiliation{Institute of Theoretical Physics and Astrophysics, National Quantum Information Centre, Faculty of
Mathematics, Physics and Informatics, University of Gda{\'n}sk, 80-308 Gda{\'n}sk, Poland}
\author{M. Studzi\'nski}
\affiliation{DAMTP, Centre for Mathematical Sciences, University of Cambridge, Cambridge~CB30WA, UK}

\maketitle

Obtaining the classical secret key for data encryption via quantum states is one of the main contributions of quantum information theory towards security in the era of information \cite{BB84}. This goal has been achieved using the celebrated maximally entangled state \cite{EPR,E91,ShorPreskill}. However, the quantum states which have a property that after the measurement one can get at least $m$ bits of a secret secure key (against quantum eavesdropper) form a much broader class of states called {\it private states} \cite{pptkey}. The maximally entangled state is an example of a private state. In general, a private
state has two subsystems: the key part ($AB$) and the shield ($A'B'$). From the key part one can obtain the secure key via von Neumann measurements on its local subsystems, while
the shield is protecting the key part from the Eavesdropper who holds the purifying system of the total state. More precisely, any private state with {\it at least} $m$ bits of key has the
form
\be
\gamma_{ABA'B'} = {1\over d}\sum_{ij} |ii\>\<jj|_{AB}\otimes U_i \sigma_{A'B'} U_j^{\dagger},
\label{eq:private}
\ee
where $\sigma_{A'B'}\equiv\sigma$ is an arbitrary state on $A'B'$ and $U_i$ are unitary
transformations. 

Private states have been used to formalize quantitative relation between secrecy and entanglement. Namely, it has been shown that the classical secure key, is in, fact an entanglement measure denoted as $K_D$ \cite{pptkey, AugusiakH2008-multi}.
This led to upper bounds on the secure content of quantum states via {\it relative entropy of entanglement} \cite{pptkey,keyhuge,AugusiakH2008-multi} and {\it squashed entanglement} \cite{Christandl-phd,Christandl-Ekert-etal,Yang2007-multi-squash,Wilde-squash} and further generalizations for quantum channels via {\it squashed entanglement of a quantum channel}
\cite{TGW-sq-channel,TGW-sq-natcom} and relative entropy of entanglement extended to quantum channel \cite{Pirandola-key-network} (see also \cite{Takeoka-SW-bound,WTB-convers-priv, Laurenza-relative,Pirandola-review} in this context).
Recently, there has been shown a related impossibility result that one cannot achieve non-negligible amount of key \cite{BCHW2015} in the framework of quantum repeaters \cite{repeatersPRL} using some certain approximate private states (a problem of quantum {\it key-repeaters}). 

 In general the importance of the class of private states follows from the fact that any quantum key distribution protocol (including the so-called quantum device independent ones \cite{BHK,BCK}) is equivalent to a protocol whose output is a private state.
Hence, in a coherent view on quantum mechanics, the output of any quantum key distribution protocol ${\cal P}$ has form of the state (close to) the one given in 
equation (\ref{eq:private}). Having $\gamma_{ABA'B'}$ as the output of ${\cal P}$. The part $AB$, measured in computational basis, yields outcomes which are directly used for one-time-pad encryption. 

By definition a private state has then directly accessible key in the key part while
the shield system $A'B'$ has been considered as a passive resource merely assuring 
security of the key part. Intuitively, however one might consider distillation
of key from also from the shield system. The caveat though is that the latter distillation
should not leak information about the key resting in the key part.
The private states which have the property, that their whole secure content is accessible via measuring their key part in computational basis, i.e. for which $K_D(\gamma_{d_k}) = \log d_k$, where $d_k$ is the local dimension of $A$ and $B$, are called {\it irreducible}. All other private states are called {\it reducible}.

The aim of this manuscript is then twofold: first, to characterize irreducible private states, and study their properties, the second, more importantly, to provide the protocol which distills key from not irreducible (i.e. {\it reducible}) private states at higher rate than the trivial one, equal to $\log d_k$. We then naturally encounter a hard problem of whether there exist entangled, but key-undistillable states, which appears to be tightly connected to the problem of characterization of irreducible private states. We therefore also study properties of the set of key-undistillable states.

Before stating the main result, we first describe an intuition, which motivates our
approach. Let us suppose, that the state $\gamma_{ABA'B'}^{\ot n}$ has been attacked by a hacker Eve i.e. measured on $AB$ part, and Eve has copied the outcomes. The leftover state is $\hat{\gamma}_{ABA'B'}^{\otimes n}= (\sum_i |ii\>\<ii|\otimes U_i\sigma U_i^{\dagger})^{\ot n}$. Intuitively, if $\gamma_{ABA'B'}$ is irreducible, Eve has learned the whole key, as
it was located in its key part. This implies, that the total state $\hat{\gamma}_{ABA'B'}^{\ot n}$ should not have distillable key. Now, having access to
the "flags" $|ii\>\<ii|_{AB}$ Alice and Bob can sort the states $\sigma_i := U_i\sigma U_i^{\dagger}$ on systems $A'B'$ and get, by the typicality argument 
$\approx {n\over d}$ copies of state ${\tilde{\sigma}}\equiv\sigma_0\ot \cdots \ot \sigma_{d-1}$. Hence the state $\tilde{\sigma}$ must be also key-undistillable, as obtained by local operations from the key-undistillable $\hat{\gamma}_{ABA'B'}$. We therefore  arrive at the intuition
that the following statement about irreducible private states is true:
\be
\gamma_{ABA'B'}\,\mbox{ is irreducible} \Rightarrow K_D(\sigma_0\ot \cdots \ot\sigma_{d-1}) = 0.
\label{eq:simple}
\ee
Negating the RHS of the above we obtain:
\be
K_D(\sigma_0\ot \cdots \ot \sigma_{d-1})>0 \Rightarrow \gamma_{ABA'B'} \mbox{ is reducible.}
\ee
Now, as the main result, we prove a much stronger, quantitative version of the above implication. Namely that for any private state $\gamma_{ABA'B'}$, in particular a reducible one, there is:
\be
K_D(\gamma_{ABA'B'})\geq \log d + {1\over d}K_D(\sigma_0\otimes \cdots \otimes \sigma_{d-1}).
\label{eq:performance1}
\ee
We prove this lower bound by providing explicitly a protocol achieving the rate of RHS of equation~\eqref{eq:performance1}.

We then ask if the proposed protocol has an optimal rate. We prove a new upper bound on it of the form:
\be
K_D(\gamma_{ABA'B'}) \leq \log d + {1\over d}E_r^{\infty}(\sigma_0\otimes \cdots \otimes \sigma_{d-1}),
\label{eq:new-bound}
\ee
where $E_r^{\infty}$ is the regularized relative entropy of entanglement (the previous upper bound~\cite{keyhuge,karol-PhD} was in terms of $\sum_i E_r^{\infty}(\sigma_i)$ which is sometimes strictly larger than RHS). We observe however, that this bound may not be tight even for irreducible private states.
This is because there may exist {\it entangled key undistillable states}. Indeed, since 
$E_r^{\infty}(\rho) >0$ iff $\rho$ is entangled i.e. this measure is {\it faithful},
the RHS can differ from LHS in that case by arbitrary amount.

The problem of whether  entangled key undistillable states exist appears to be non-trivial. 
We do not judge if the answer is "yes", such states exist or rather "no", they do not. 
In what follows, much in the spirit of the community of computer science approaching the $P\neq NP$ problem, we check what are the possible consequences of the two answers. Namely, if the answer is "no" there are no entangled key undistillable states, then via equation (\ref{eq:performance1}) we have fully characterized the set of irreducible private states, and the bound (\ref{eq:new-bound}) is tight. If however it is not the case, plenty of questions are ready to be asked.

We study consequences of the mentioned alternative from the perspective of the upper-bounds on distillable key of a private state in terms of measures of entanglement.
If the world is such that entangled key-undistillable states {\it do not} exist, then as we noted, the only irreducible private states are strictly irreducible, which in fact satisfy $K_D(\gamma)=E_r(\gamma)$, see Figure~\ref{sets}, where $E_r$ is the relative entropy of entanglement \cite{VPRK1997}. If however, the world is reach enough, so that such states exist, there must exists a private state which is not strictly irreducible, that is outside of the set of states satisfying $K_D=E_r=E_r^{\infty}$.

Apart from studying consequences of (non)existence of entangled key undistillable states, we study the properties of the set of key-undistillable states $\mathcal{Z}$ in general.
It is known, that this set contains all {\it separable} i.e. disentangled states,
as they can be made from product states (clearly having $K_D=0$) by public discussion \cite{Curty2004}.
Separable states are far from the private states in the norm distance; for any private state, with the dimension of the shield part $d_s$, it holds that~\cite{keyhuge,karol-PhD}:
\be
\label{distSep1}
\sigma \in \mathcal{SEP}^{(d_k,d_s)}\, \Rightarrow\, ||\sigma - \gamma||_{1} \geq 1-\frac{1}{\sqrt{d_k}},
\ee
where by $\mathcal{SEP}^{(d_k,d_s)}$ we denote the set of separable states in the cut $AA':BB'$, and for any $X\in \mathbb{M}(d,\mathbb{C})$ by $||X||_1\equiv \tr \sqrt{XX^{\dagger}}$ we define its trace norm.
What we are further able to show, is an analogous statement for states from ${\cal Z}$: 
\be
\sigma \in {\cal Z} \Rightarrow\, ||\sigma - \gamma||_{1} \geq {1\over 6} + {2\over 3 \log d}.
\ee
It is much weaker than that from equation~\eqref{distSep1}, but we pay a price for generality: we do not know the structure 
of all key undistillable states, and in derivation we can use only the fact that $\sigma$ is key-undistillable.

\begin{figure}[h!]	
	\includegraphics[trim={0 10cm 0 0},width=0.55\textwidth]{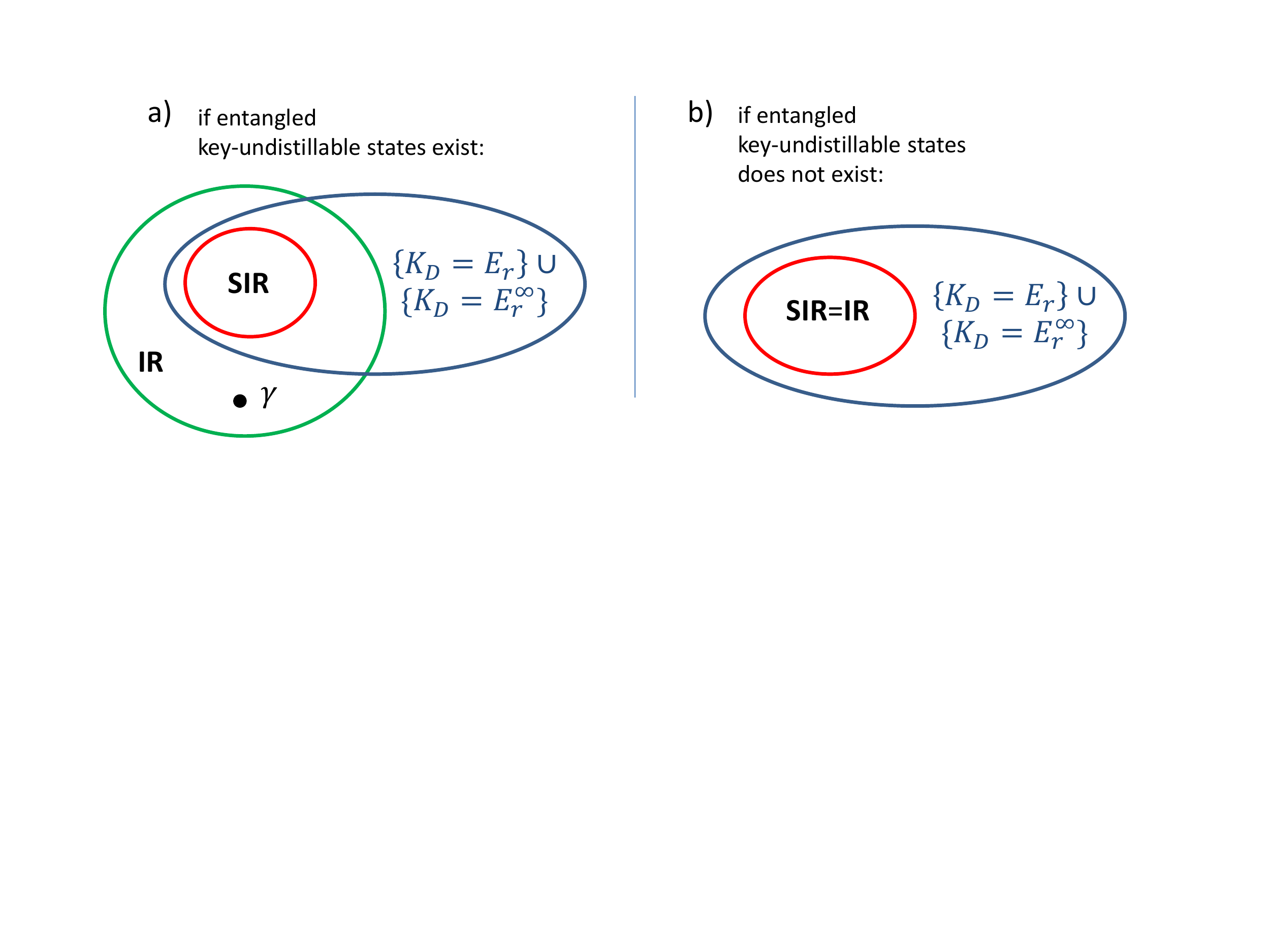}
	\caption{It is a difficult open question if entangled  key undistilable states exist. Given answer to this question we follow what we will know about irreducible private states. On panel a) we present the situation when entangled zero key states exist. Then the set of strictly irreducible private states is a proper subset of all irreducible private states. There also exists then a private state $\gamma$ which is irreducible and is beyond the set of states $\{K_D=E_r\}\cup \{K_D=E_r^{\infty}\}$. On panel b) we present the situation when entangled  key-undistillable states do not exist. This assumption implies that the set of strictly irreducible private states equals to the set of irreducible private states and as a consequence we have their full characterization in this manuscript.}
	\label{sets}
\end{figure}

Let us note here, that the class of irreducible private states has not been characterized so far, except of its one subclass  called in \cite{Ferrara} as {\it strictly irreducible} states. Strictly irreducible states \cite{KH-phd} are the private states satisfying that diagonal blocks of the matrix of a private states, $\sigma_i \equiv U_i\sigma U_i^{\dagger}$ in Eq. (\ref{eq:private}) are separable for all $i \in \{0,\ldots,d-1\}$. These states appeared
recently in the problem of quantum key-repeaters. We therefore discuss also possible further applications of our main result in this context.
Namely M. Christandl and R. Ferrara have shown in \cite{Ferrara} that strictly irreducible private states have the amount of one-way key-repeaters rate bounded from above by one-way distillable entanglement of $\gamma\otimes \gamma$. It would be interesting to extend this result to {\it all} irreducible private states, and present the answer in terms of the amount of key achieved via key-repeaters (key-repeater rate) and distillable entanglement of $\gamma\otimes \gamma$. Such result would answer the question of the role of private states in key-repeaters. The first important step in this direction, that we have done is proving that all key undistillable states are bounded away in trace norm from any private state, since the distance between separable and private states is crucial to the findings of \cite{Ferrara}. This problem seems to be hard, as there is no evidence that there even exist states with zero distillable key. 

The paper is organized as follows. Section \ref{NaD} provides definitions that are further used on. Section \ref{sec:protocol} contains the main result, which is a tighter lower bound on distillable key of private states. We present there a protocol of key distillation which uses both the key part and the shield of a private state. We also give an upper bound on its rate. Section \ref{sec:properties} is devoted
to some properties of the set of irreducible private states, including its partial characterization. In Section \ref{sec:prop-kundist} we provide a lower bound on distance of key-undistillable states from private states via the necessary condition of key-undistillability as well as show its applications to single copy quantum key repeaters.
In Section \ref{sec:approx} we study approximate irreducible states in the light of considerations of Section \ref{sec:properties}.
The last section is devoted to the discussion of the above results.


\section{Private states and irreducible private states}
\label{NaD}
In this section we introduce rigorous definitions of private states as well as irreducible and strictly irreducible private states which are crucial in further parts of this manuscript.  One example of a private state is a \textit{basic private state} $\rho_{ABA'B'}$, acting on a Hilbert space $\mathcal{C}^{d_A}\ot \mathcal{C}^{d_B}\ot \mathcal{C}^{d_{A'}}\ot \mathcal{C}^{d_{B'}}$. A basic pdit is of the form
\be
\label{basic-pdit}
\rho_{ABA'B'}\equiv P^+_{AB}\ot \sigma_{A'B'},
\ee
where $P^+_{AB}$ is a maximally entangled state between on $\mathcal{C}^{d_A}\ot \mathcal{C}^{d_B}$, and $\sigma_{A'B'}$ is an arbitrary bipartitie quantum state on $\mathcal{C}^{d_{A'}}\ot \mathcal{C}^{d_{B'}}$. Imposing now separability of $\sigma_{A'B'}$ we obtain an example of irreducibly private state. It is not hard to see that in this case we can extract the secret key only from the singlet $P^+_{AB}$. Of course both classes introduced above are of much more richer structure and below we present their rigorous definitions. Let us start from the definition of the private state (pdit) \cite{keyhuge}:
\begin{definition}[of private state]
	\label{pbit}
A state $\rho_{ABA'B'}$ on a Hilbert space $\mathcal{C}^{d_A}\ot \mathcal{C}^{d_B}\ot \mathcal{C}^{d_{A'}}\ot \mathcal{C}^{d_{B'}}$ with dimensions $d_A=d_B=d$, $d_{A'}$ and $d_{B'}$, of the form
\be
\label{pdit}
\gamma_d\equiv\frac{1}{d}\sum_{i,j=0}^{d-1}|e_if_i\>\<e_jf_j|_{AB}\ot U_i\sigma_{A'B'}U_j^{\dagger},
\ee
where the state $\sigma_{A'B'}$ is an arbitrary state of subsystem $A'B'$, $U_i'$s are arbitrary unitary transformations and $\left\lbrace|e_i\> \right\rbrace_{i=0}^{d-1}, \left\lbrace |f_j\>\right\rbrace_{j=0}^{d-1}$ are local bases on $A$ and $B$ respectively, is called the private state or pdit. In the case of $d=2$, the state in (\ref{pdit}) is called - pbit.  Further in the text we will use interchangeably the notation of the private state, i.e.  $\gamma_d,\gamma_{ABA'B'}$ or simply $\gamma$, and sometimes we will drop lower indices when there is no danger of confusion.
\end{definition}
In what follows we will set $\{|e_i\>\ot |f_j\>\}$ basis to be computational for simplicity, while all the facts hold for a general case of
arbitrary product basis of a pdit.
The part $AB$ of a pdit is called  {\it the key part}, while the subsystem $A'B'$ its {\it shield}. For any pdit $\gamma_d\in \mathcal{C}^{d}\ot \mathcal{C}^{d}\ot \mathcal{C}^{d_{A'}}\ot \mathcal{C}^{d_{B'}}$, which is secure in the standard basis by $\sigma_i$ for $i=0,\ldots,d-1$ we denote states which appear on the shield of the pdit, after obtaining an outcome $|ii\>_{AB}$ in the measurement performed in the standard basis on its key part. We call the states $\sigma_i$ the {\it conditional states}.

We now come to the central definition to all our further considerations:
\begin{definition} Private state $\gamma_{d,d'}$ is called {\it irreducible} iff $K_D(\gamma_{d,d'})=\log d$. The set of irreducible private states with local dimension of key part $d$ and total dimension of the shield part $d'$, will be denoted as $\mathcal{IR}_{d,d'}$.
\end{definition}

We will sometimes suppress the dimension subscript from notation in $\mathcal{IR}_{d,d'}$.
Having the private state $\gamma$ as in equation~\eqref{pdit}, then the state
\be
\label{key-attacked}
\hat{\gamma}:=\frac{1}{d}\sum_{i=0}^{d-1}|ii\>\<ii|\ot U_i\sigma U_i^{\dagger}
\ee
is called the \textit{key-attacked private state}. It is a private state which has been measured on its key part $AB$.
Now we are ready to introduce important subset of the $\mathcal{IR}$ is the set of strictly irreducible private states $\mathcal{SIR}$. Denoting by $\mathcal{SEP}$ set of bipartite separable states we have the following:
\begin{definition}\cite{Ferrara}
We say a private state $\gamma$ is strictly irreducible if the state
	$\hat{\gamma}=\frac{1}{d}\sum_i |ii\>\<ii| \ot U_i\sigma U_i^{\dagger}$
	 if consists of  separable states with respect to the systems on the shield part on the diagonal, i.e. $U_i\sigma U_i \in \mathcal{SEP}$. We denote such states as $\<\gamma \>$ or $\gamma^{\<d\>}$.
\end{definition}
One of the most natural example of strictly irreducible private state is when we chose in the above definition $\sigma=\mathbf{1}$, where $\mathbf{1}$ is an identity operator on the shield part of $\gamma$. Next,  a less trivial example of state within the set $\mathcal{SIR}$ is the "flower state" $\gamma_{flower}$~\cite{karol-PhD}, which was shown~\cite{2005PhRvL..94t0501H} to lock entanglement cost. We have that $\gamma_{flower}\in \mathcal{C}^{2}\ot \mathcal{C}^{2} \ot \mathcal{C}^{d^2}\ot \mathcal{C}^{d^2}$ is of the form:
\be
\label{flower}
\gamma_{flower}\equiv \frac{1}{2}\begin{pmatrix}
	\sigma & 0 & 0 & \frac{1}{d}U^T\\
	0 & 0 & 0 & 0\\
	0 & 0 & 0 & 0\\
	\frac{1}{d}U^* & 0 & 0 & \sigma
\end{pmatrix},
\ee
where $\sigma=\frac{1}{d}\sum_{i=0}^{d-1}|ii\>\<ii|$ is maximally correlated state, and $U=\sum_{i,j=0}^{d-1}w_{ij}|ii\>\<jj|$ is the embedding of unitary transformation $W=\sum_{i,j=0}^{d-1}w_{ij}|i\>\<j|=H^{\ot \operatorname{log d}}$ with $H$ being Hadamard transformation.

\section{The main result: Protocol of distilling key from private states}
\label{sec:protocol}
Here, we will present the protocol which enables the characterization of irreducible private states. As announced, it can distill the key not only from the key part but also from the shield of a private state in the case where $\sigma_0\ot \cdots \ot\sigma_{d-1}$ is key distillable. We first show its rate, which is the main result of this section and the core of our further reasoning in this manuscript. Here we state the result for bipartite private states, however analogues result
holds as well for the case of multipartite private states \cite{AugusiakH2008-multi}, as we show in Appendix \ref{sec:c}.

The distillable key can be lower bounded as follows:
\begin{theorem}
\label{th:main} For a private state $\gamma_d$, there is
\be
\label{rate-thm}
K_D(\gamma_d) \geq \log d + \frac{1}{d}K_D(\sigma_0 \ot \sigma_1 \ot \cdots \ot \sigma_{d-1}).
\ee
\end{theorem}

We achieve the above result via the following protocol defined by parameters $m$ and $k$. The input to the protocol $P(m,k)$ is $\gamma^{\ot (m\times k)}$.  This protocol depends on a threshold $\delta(m)$ which approaches $0$ with increasing $m$. We fix the value of this threshold later, when study the rate of this protocol.
\begin{enumerate}
\item For each of $k$ blocks of $m$ states Alice and Bob repeat items $2$-$6$ as follows:
\item Alice and Bob measure their subsystems of the key part $A$ and $B$ in the secure basis of $\gamma$.
\label{item:measurement}
\item If number $t_s(m)$ of any symbol $s\in \{0,...,d-1\}$ in the key $|i_0\>\ot \cdots \ot|i_{m}\>$ with $i_j\in\{0,...,d-1\}$ is below the threshold $\delta(m)$, i.e. $|{m\over d} - t_s(m)|>\delta(m)$, they trace out the original state, produce the error state $|e\>_{AA'}\ot|e\>_{BB'}$. They then proceed with step $2$ on the next block if such is left, or go directly to step $7$ otherwise. If however $|{m\over d} -t_s(m)|\leq  \delta(m)$ for all $s$, they proceed with the steps $4-6$ of the protocol.
\label{item:ab}
\item Alice and Bob perform each control-permutation operation: depending on the value of the key $|i_0\>\ot\cdots\ot|i_{m}\>$ with $i_j\in\{0,\ldots,d-1\}$
they permute states on the systems $A'B'$ from its original form $\sigma_{i_0}\ot \cdots\ot\sigma_{i_{m}}$ to $\sigma_0^{\ot t_0(m)}\ot \cdots\ot\sigma_{d-1}^{\ot  t_{d-1}(m)}$.
\label{item:permute}
\item If necessary Alice and Bob trace out some of the states $\sigma_0,\ldots,\sigma_d$ leaving only $t_{min}(m)=\lfloor \frac{m}{d} - \delta(m)\rfloor$ of copies of the state $\widetilde{\sigma} \equiv \sigma_0\ot \cdots \ot\sigma_{d-1}$.
\label{item:trace}
\item As a subprotocol, Alice and Bob perform on systems $A'B'$ any optimal key distillation protocol of the state $\widetilde{\sigma}$, acting on $\widetilde{\sigma}^{\ot t_{min}(m)}$, with output close in the trace norm
to the state $\gamma_{A'B'A''B''}$. The state on systems $AB$ from the item \ref{item:ab1} and state on systems $A'B'A''B''$ form one of the $k$ input states for the final step $7$ of the protocol.
\item On $k$ output states of the $k$ repetitions of items $2$-$6$ (or items $2-3$ if unsuccessful in item $3$ respectively) Alice and Bob perform the Devetak-Winter protocol \cite{DevetakWinter-hash}.

\end{enumerate}

\begin{proof}[Sketch of the proof]
The detailed proof of Theorem \ref{th:main} is given in Appendix~\ref{app:A}.We also refer reader to Figure~\ref{Protocol} on which we present the performance of our protocol  for the finite and fixed number of input states. We provide here briefly the idea.  First let us note, that if the state $\widetilde{\sigma} =\sigma_0\ot \cdots \ot\sigma_{d-1}$ is key-undistillable, then there is obvious protocol of distilling $\log d $ of key from $\gamma_d$ : just via measuring its
key part. Hence, in case $K_D(\widetilde{\sigma})=0$, our bound is trivially satisfied.  Thus, in what follows we will assume that $K_D(\widetilde{\sigma}) > 0$.
It is intuitive that given the key from key part has been used, there are
remaining states on the shield part.  The major obstacle is that these two parts are correlated. We overcome this problem in steps $2-4$.

 In the first step (item \ref{item:measurement}) the parties get to know the value of the key on key part. The next step (item \ref{item:ab}) shows that our protocol is of the kind "all versus noting": if the key has atypical value (having too small number of some of the symbols $0,\ldots,d-1$), the parties trace out the whole state, create the error flags and proceed with step $\ref{item:measurement}$ for the next of $k$ blocks. This however occurs extremely rarely due to concentration property of independent, identically distributed random variables. Hence with probability almost one, the key value is typical, and Alice and Bob proceed with the main step of the protocol (item \ref{item:permute}).
	
 In the fourth step they decouple the key part and the shield by sorting conditional states on the shield part: all states labelled with $0$ together, up to $d-1$. The only correlations which survive between the system $AB$ and $A'B'$ are due to the type that is how many symbols $0$'s, $1$'s up to $d-1$'s are in the value of the key on system $AB$.

 The next step (item \ref{item:trace}) is provided for technical reasons, as it allows for easy calculation of the rate of the protocol. In the pre-last step Alice and Bob perform
a key distillation protocol on $t_{min}(m)$ copies of the state $\widetilde{\sigma}=\sigma_0\ot \cdots \ot\sigma_{d-1}$ which, by definition, ends up with a state close to some private state $\gamma_{A'B'A''B''}$. The state on $ABA'B'A''B''$ has now the
perfect key on $A'B'$ and the partially secure key on $AB$ about which Eve knows the type that is a number of symbols in the string of the value of this key.
Fortunately, this knowledge has very low entropy, sub-linear in $m$, and will not affect the total rate of the protocol.

 After repeating steps $2-6$ $k$ times (or items $2-3$ if unsuccessful in item $3$ respectively), in the last step Alice and Bob perform the Devetak-Winter protocol on the subsystems $AA'BB'$ of $k$ such obtained states  \cite{DevetakWinter-hash}. It then remains to check that the rate of the Devetak-Winter protocol reads in this case approximately $r(m,k)=k(m \log d - (m+d-1)\,h(\frac{m}{m+d-1}) + t_{min}(m)K_D(\widetilde{\sigma})-o(m))- o(k)$. The first term corresponds to $m\log d$ bits of initial ideal key on $AB$
however it is lowered by the second term $(m+d-1)\,h(\frac{m}{m+d-1})$. The latter corresponds to Eve's knowledge of the type. Indeed, there are no more than
\be
{m+d-1\choose m} \leq 2^{(m+d-1)h\left( \frac{m}{m+d-1}\right) }
\ee
of different types of strings of length $m$ of symbols from $d$-ary alphabet \cite{Renner-Phd}. The third term $t_{min}(m)K_D(\widetilde{\sigma})$ is the rate
of the protocol of key distillation from $t_{min}(m)$ copies of state $\widetilde{\sigma}$.
Hence, choosing properly value of $\delta(m)$  approaching $0$ with $m$ going to infinity (see Appendix~\ref{app:A}), we obtain that the rate of our protocol which is $\frac{1}{m \, k} r(m,k)$ reaches asymptotically (taking $m,k\rightarrow \infty$) value $\log d + {1\over d} K_D(\sigma_0\ot \cdots \ot\sigma_{d-1})$, as claimed.
\end{proof}

\begin{figure}[h!]	
		\includegraphics[width=0.5\textwidth]{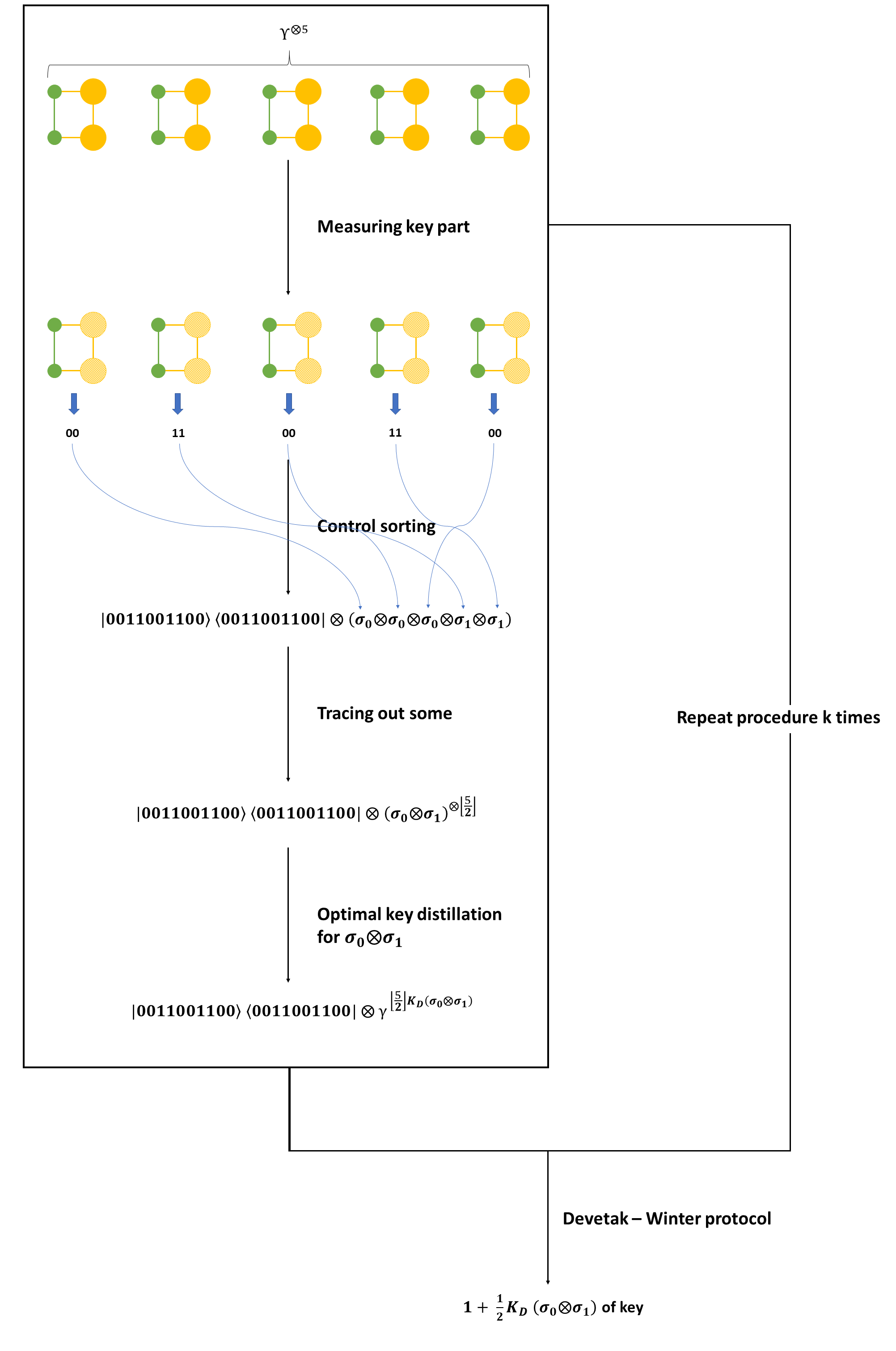}
		\caption{On this figure we present performance of the protocol from Theorem~\ref{th:main} for $n=5$ copies of the private pbit $\gamma$. Single private state $\gamma$ is represented as a collection of four dots. By smaller green dots we denote qubits forming the key part of pbit and by larger yellow its respective shield part. At the beginning Alice and Bob measure their subsystems of the key part $A$ and $B$ in the secure basis of $\gamma$ obtaining respective bit strings $00,11,00,11,00$. Assuming that parties can proceed (3rd point from Theorem~\ref{th:main} is not fulfilled.), they perform control sorting depending on the value of the key by permuting states on the parts $A'B'$ obtaining $\sigma_0\ot \sigma_0\ot \sigma_0 \ot \sigma_1\ot \sigma_1$, and eventually tracing out some of their subsystems. Having that, parties perform on systems $A'B'$ any optimal key distillation protocol. Finally, after $k$ repetitions of items 2-6 from the protocol, Alice and Bob perform the Devetak-Winter protocol on systems $AA'BB'$ obtaining $1+\frac{1}{2}K_D(\sigma_0\ot \sigma_1)$ of key.}
		\label{Protocol}
\end{figure}

From  theorem~\ref{th:main} we have immediate important corollary  which is in fact a necessary condition for irreducibility:

\begin{corollary}
For any irreducible private state $\gamma$ of the form
\be
\frac{1}{d} \sum_{ij} |ii\>\<jj| \ot U_i \sigma U_j
\ee
there is
\be
K_D(\sigma_0\otimes \cdots \otimes \sigma_{d-1})=0.
\ee
In particular we have
\be
\forall i \in \{0,\ldots,d-1\} \quad K_D(\sigma_i) = 0,
\ee
where $\sigma_i = U_i \sigma U_i^{\dagger}$. 
\end{corollary}

We note here that our protocol needs for a pdit $\gamma_d$ by definition a huge number of its copies: $ks$ for some parameter $s$, $s$ and $k$ increasing asymptotically. It is really huge, in comparison with the fact that one can get $\log d$ of key from a {\it single} copy of $\gamma_d$ by measuring its key part. However, it shows what can be done in general, asymptotically, which allows to study the distillable key of reducible private states: so that it works, this number is larger than $d^{3/2}\log d'$.

\vspace{0.5cm}
It is tempting to ask if the protocol that we proposed is optimal. In what follows we will improve
the bounds on relative entropy of entanglement for private states which will show to what extend
performance of our protocol is not tight. So far the following bound was known \cite{keyhuge,karol-PhD}:
\be
K_D(\gamma) \leq E_r^{\infty}(\gamma) \leq \log d + {1\over d}\sum_i E^{\infty}_r(\sigma_i).
\ee
In the next theorem we present a bit tighter bound. The fact that it is tighter comes from
the subadditivity of $E_r$. The proof of it is directly inspired by the protocol from the above section.

\begin{theorem}
For any pdit $\gamma_{d,d'} $ we have
\be
\label{Er1}
E_r^{\infty}\left(\gamma_{ABA'B'}\right) \leq \operatorname{log}d+\frac{1}{d}E_r^{\infty}\left(\widetilde{\sigma}\right),
\ee
where $\widetilde{\sigma}=\sigma_0\ot \sigma_1\ot \cdots \ot \sigma_{d-1}$ and $\sigma_i$ are conditional states on $A'B'$ for $i\in \left\lbrace 0,1,\ldots,d-1\right\rbrace $.
\label{thm:upper-bound-er}
\end{theorem}

\section{Properties of the set of irreducible private states}
\label{sec:properties}
Recall that by $\mathcal{IR}_{d,d'}$ we denote the set of irreducible private states with the dimension of key $d \ot d$
and that of shield $d' \ot d'$. Not to overload notation, we will omit the dimensions writing $\mathcal{IR}$ if some fact holds for any fixed dimensions or the dimensions are known from the context. It is natural to imagine that the key obtained from two different states from $\mathcal{IR}$
is no greater than the sum of the keys from each of them. In what follows we will prove this for certain subsets of $\mathcal{IR}$. 

We are now going to show an important example where the above observation applies,  firstly by formulating the following:
\begin{observation}
	The following sets are closed under the tensor product.
\ben
\mathcal{R} \equiv \{ \rho : K_D(\rho) = E_r(\rho)\}, \\
\mathcal{R}_{\infty} \equiv \{\rho : K_D(\rho) = E_r^{\infty}(\rho)\}, \\
\mathcal{I} \equiv \{ \rho : K_D(\rho) = I_{sq}(\rho)\},
\een
where $E_r$ is the relative entropy of entanglement, $K_D$ is the distillable key, and $I_{sq}$ is the squashed entanglement measure.
\label{closed:sets}
\end{observation}

Because the tensor product of two private states is again a private state (up to a local unitary transformation \cite{Ferrara}), by the above
observation as an immediate consequence we have that $\mathcal{R}_{ir}$, $\mathcal{R}_{ir,\infty}$, and $\mathcal{I}_{ir}$ which are intersections of the set of irreducibly private states $\mathcal{IR}$ with respective sets from Observation~\ref{closed:sets}, are closed under the tensor product.

Moreover, the above sets are obviously non-empty, and have a non-empty intersections. Indeed, let us fix dimensions and consider:
$\mathcal{R}_{ir}(d,d') = \mathcal{IR}_{d,d'}\cap \mathcal{R}({d,d'})$, $\mathcal{R}_{ir,\infty} = \mathcal{IR}_{d,d'}\cap \mathcal{R}_{\infty}({d,d'})$ and $\mathcal{I}_{ir}(d,d') = \mathcal{IR}_{d,d'}\cap \mathcal{I}({d,d'})$.
\be
|\psi_+\>\otimes \sigma_d' \in \mathcal{R}_{ir}(d,d') \cap \mathcal{R}_{ir,\infty}(d,d') \cap \mathcal{I}_{ir}(d,d')
\ee
for any dimension $d'$, where $|\psi_+\> = {1\over \sqrt{d}}\sum_i |ii\>$, that is a maximally entangled state  and
$\sigma_d'$ is an arbitrary $d'^2$-dimensional separable state. In other words, all three sets have in their mutual intersection
the set of basic private state with a separable shield.

Let us note here, that there is $\mathcal{R}_{ir}(2,d') \neq \mathcal{I}_{ir}(2,d')$ for any dimension $d'$. This is because the class of states "flower state" $\gamma_f$
which are pbits with $d$ dimensional shield satisfy $E_r(\gamma_f) = 1$ but $I_{sq}(\gamma_f) = 1 + {\log d \over 2} $ \cite{ChristandlAs}.

\subsection{Irreducible private states vs existence of bound key states}
We say that a quantum state $\rho$ has the bound distillable key if $K_D(\rho) = 0$ and $\rho$ is entangled. It is a widely
open problem if such states exist. If so, entanglement and secrecy would be different resources. We are far from solving it;
we rather connect this with the problem of studying the irreducible private states. 
One might think, that existence of bound-key states
	implies that some irreducible private states are beyond the set $\mathcal{R}$
	trivially (see Observation~\ref{closed:sets}). Indeed, if bound key state exists, say $\rho_{bk}$,
	then a basic pdit $|\psi_+\>\<\psi_+|\otimes \rho_{bk}$ is example of such a state. There are however two issues:
	firstly it is not clear what is the lower bound on $E_r$ in this
	case, as this measure is subadditive. Secondly, and more importantly,
	it is not clear that $K_D$ of this state equals $\log d$. This is
	because, in principle, the singlet state could act as a catalyst
	unblocking the key of $\rho_{bk}$ (key is by definition superadditive).
	Therefore it is not immediate to construct an {\it irreducible}
	private state outside of the set $\mathcal{R}$ which is irreducible, basing
	solely on the fact that there exists some bound key state, and
	even not easy to show that it exists.
	In what follows, however, we show that such a state must exists
	in that case. More precisely, if there are no bound key states,
we have the full characterization of the set $\mathcal{IR}$:

Clearly, we have $\mathcal{R}_{ir}\subseteq \mathcal{IR}$. We consider now the condition under which the converse inclusion holds.

\begin{proposition}
	The set $\mathcal{R}_{ir}$ does not equal  set of irreducibly private states if and only if there exists a state with the bound key. The same holds
for the set $\mathcal{R}_{ir,\infty}$ in place of $\mathcal{R}_{ir}$.
\label{prop:charac}
\end{proposition}

We have then immediate corollary characterizing the set of irreducible private states of arbitrary dimension of the key part $d$ and qubit dimension $d'=2$ of the shield. Additionally, in the particular case of qubit key part, $d=2$ we are able to calculate log-negativity of these states. 
\begin{corollary}
	\label{cor4}
The set of irreducible private bits $\mathcal{IR}_{d,2}$ consists of the strictly irreducible pdits, that is those states
$\gamma_d = {1\over d} \sum_{i=0}^{d-1}|ii\>\<ii|\ot U_i\sigma U_i^{\dagger}$ for which
there is:
\be
\label{sigmas}
U_0\sigma U_0^{\dagger}, \ldots,\  U_{d-1}\sigma U_{d-1}^{\dagger} \in \mathcal{SEP}.
\ee
Moreover, for $d=2$ with $X= U_0\sigma U_1^{\dagger}$, there is $E_N(\gamma_2) = \log \left( 1 + ||X^{\Gamma}||_1\right) $ where $\Gamma$ denotes the partial transposition
and $E_N$ is the log-negativity entanglement measure.
\end{corollary}

The above corollary in particular  holds for the strictly irreducible pbits, considered in \cite{Ferrara}.
We note that the form of the matrix in (\ref{cor4}) implies that all diagonal blocks of the matrix of a private bit from $\mathcal{IR}_{2,2}$ are separable.

\begin{rem}
If Proposition~\ref{prop:charac} was true for some entanglement measure $M$ and then was a state $\gamma \in \mathcal{R}_M\setminus \mathcal{R}_{ir}$ then it would imply existence of bound key state.  Although we know that $\mathcal{I}_{ir} \neq \mathcal{R}_{ir}$, in the case of $I_{sq}$ the analogue of Proposition \ref{prop:charac} is not true in general as the analogue of the bound (\ref{eq:old-er-bound}) does not hold e.g. for the mentioned flower state $\gamma_f$.
\end{rem}

To ensure that the state $\sigma$ from (\ref{sigmas}) is separable after the action of an arbitrary unitary operations $\{U_0, U_1,\ldots, U_{d-1}\}$ is enough to take $\sigma$ from the set of the absolutely separable states $\mathcal{ASEP}$ \cite{Kus,Gur}. Such an idea of the construction was proposed in \cite{KH-phd}, but without an explicit presentation of allowed classes of states.  Here we construct an explicit class of above mentioned states in qubit-qubit case and we show by counterexample, which is different that previously known "flower state" $\gamma_{flower}$~\cite{,karol-PhD,2005PhRvL..94t0501H}, that such class does not saturate all possible choices. 

The problem of the absolute separability or the separability from the
spectrum \cite{PhysRevA.58.883,Knill1,Moor,John} asks for a characterization of the states $\rho\in \CC^{ d}\otimes \CC^{ d}$, such that $U\rho U^{\dagger}$ is separable for all unitary matrices $U$. We know that such full characterization was done for the qubit-qubit~\cite{Moor} case and then it was generalized for the qubit-qudit case \cite{John}. Here, as we mentioned, we restrict ourselves to the case of the two qubits by recalling the following theorem from~\cite{Moor}:
\begin{theorem}
	\label{asep}
	Consider the state $\rho\in \CC^{ 2}\otimes \CC^{ 2}$. Let $\operatorname{spec}^{\downarrow}(\rho)=\{\lambda_{i} \ : \ i=1,\ldots,4\}$
	be a set of eigenvalues of $\rho$ in a non-increasing order,
	then the state is absolutely separable if and only if its eigenvalues
	satisfy 	
	\begin{equation}
	\lambda_{1}\leq\lambda_{3}+2\sqrt{\lambda_{2}\lambda_{4}}.
	\label{eq:abssep22}
	\end{equation}
\end{theorem}
We see that thanks to (\ref{asep}) by putting states $\sigma$ from (\ref{cor4}) satisfying (\ref{eq:abssep22}) we always get an irreducible pbit. Private pbits obtained by twisting of absolutely separable state together with a singlet form the set $\mathcal{TAS}$, see Figure~\ref{sets2}. Of course we can find a state $\omega \notin \mathcal{ASEP}$ and non-local twisting operations for which conditions in (\ref{sigmas}) are fulfilled. We illustrate this by the following example:
\begin{example}
	\label{ex1}
	Let us consider a separable state
	\be
	\label{omega}
	\omega=\frac{1}{4}\left(\begin{array}{cccc}
	1 & \cdot & \cdot & 1\\
	\cdot & 1 & 1 & \cdot\\
	\cdot & 1 & 1 & \cdot\\
	1 & \cdot & \cdot & 1
	\end{array}\right),
	\ee
	which is not absolutely separable, since the condition (\ref{eq:abssep22}) does not hold. This means that there exists a unitary transformation
	$U$ for which $U\omega U^{\dagger}$ is entangled, so $\omega^{\Gamma}=\left(U\rho U^{\dagger}\right)^{\Gamma}\ngeq0$, where $\Gamma$ denotes the operation of the partial transposition with respect to one of the subsystems~\cite{SepHHH,SepAP}. On the other hand, there exist non-local unitary operations $V$, such that condition $\omega^{\Gamma}=\left(V\omega V^{\dagger}\right)^{\Gamma}\geq0$ holds - this means that the state $V\omega V^{\dagger}$ is separable. Here we provide the family
	of such operations $V$
	\be
	\label{twist}
	V=\left(\begin{array}{cccc}
	\cdot & \cdot & 1 & \cdot\\
	s & -c & \cdot & \cdot\\
	c & s & \cdot & \cdot\\
	\cdot & \cdot & \cdot & -1
	\end{array}\right),
	\ee
	where $c=\cos(\theta)$ and $s=\sin(\theta)$ for $\theta \in \left[0,2\pi \right]$. It is easy to see that $\omega^{\Gamma}=\left(V\omega V^{\dagger}\right)^{\Gamma}\geq0$, which means that $V\omega V^{\dagger}\in \mathcal{SEP}$. The unitary transformation $V$ is indeed non-local, since after an action on the separable state
	\be
	\widetilde{\omega}=\frac{1}{2}\left(\begin{array}{cccc}
	1 & \cdot & \cdot & \cdot\\
	\cdot & \cdot & \cdot & \cdot\\
	\cdot & \cdot & \cdot & \cdot\\
	\cdot & \cdot & \cdot & 1
	\end{array}\right)
	\ee
	we get $\left(V\widetilde{\omega} V^{\dagger}\right)^{\Gamma}\ngeq 0$. Private states $\gamma(\omega)$ constructed by the use of the state $\omega$ form equation~\eqref{omega} and twisting operations from~\eqref{twist}  are outside of the set $\mathcal{TAS}$ (twisted absolutely separable states), but still within the set of strictly irreducible pbits $\mathcal{SIR}$, see Figure~\ref{sets2}.
\end{example}
	\begin{figure}[h!]	
		\includegraphics[width=0.25\textwidth]{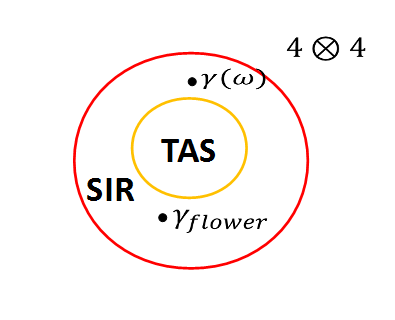}
		\caption{Within the set of all strictly irreducible private states of dimension $4 \otimes 4$ we distinguish the set $\mathcal{TAS}$ of twisted absolutely separable private states. States from $\mathcal{TAS}$ are obtained by twisting a singlet together with an absolutely separable state on the shield part. We see that outside of $\mathcal{TAS}$ but still within the set of strictly irreducible private states $\mathcal{SIR}$ we can find the "flower state" $\gamma_{flower}$ as well as class of the private states $\gamma(\omega)$ constructed in Example~\ref{ex1}.}
		\label{sets2}
	\end{figure}

\section{Properties of the key-undistillable states}
\label{sec:prop-kundist}
In this section for the first time we investigate properties of the conditional states of the irreducible private states without presenting their explicit form. For the key-attacked state of an irreducible private state we are able to show even more, namely their key-undistillability. We then show a  lower bound on the trace distance between the key-undistillable and private states.

\begin{observation}
	The key-attacked state $\hat{\gamma}$ of an irreducible private state $\gamma$
is key-undistillable. The same holds for $\hat{\gamma}^{\otimes n}$ for $n \geq 2$.
\label{obs-keyun}
\end{observation}
Having an operation of mapping a bipartite state on systems $AB$ to a classical-classical-quantum (ccq) state on systems $ABE$ from \cite{KH-phd} (or see Section~\ref{sec:ccq} of Appendix~\ref{app:cdk} )
we are ready now for an attempt to consider a general question, interesting on its own: {\it how far are the key-undistillable states from the private states}? We base on the proof of a lower bound on distillable key given in \cite{smallkey} (see \cite{KH-phd} for more elaborative explanation):

\begin{theorem} \cite{smallkey, KH-phd} For any bipartite state $\rho_{AB}$ there is:
\be
K_D(\rho_{AB}) \geq C_{DW}\left(  (\rho_{psq})_{ccq}\right),
\ee
where $\rho_{psq}$ is the privacy squeezed state of the state $\rho$.
\end{theorem}

We will now generalize the above theorem to obtain the {\it necessary condition of key-undistillability}, which in fact can be viewed
as non-one-way-key-distillability. In what follows by $\zcal$ we will denote the set of key-undistillable states, where the dimension
is assumed to be understood from the context.

\begin{theorem}	
 For any key-undistillable state $\rho \in \zcal$, and for any local unitary transformations $U_1$ and $U_2$ which act on $\rho$ 
so that $U_1\ot U_2 \rho U_1^{\dagger}\ot U_2^{\dagger} =: \rho_{ABA'B'}$, as well as
any unitary transformation of the form $U=\sum_{ij} |ij\>\<ij|_{AB}\ot U^{(ij)}_{A'B'}$ there is:
\be
C_{DW}\left( \left( \tr_{A'B'}U\rho_{ABA'B'} U^{\dagger}\right) _{ccq}\right)  \leq 0,
\ee
where the ccq state is w.r.t. to the computational basis $\{|ij\>\}$.
\label{thm:nec-ku}
\end{theorem}

The above theorem will allow us to prove a lower bound on distance between key-undistillable states and private states, improving the result of \cite{KH-phd} where a dual fact is proved, namely that the states closer to private bits by $\delta < 0.001$ are key distillable.

\begin{theorem}
	For any state which has $K_D(\rho)=0$, and any private state $\gamma$ with the $d\otimes d$ dimensional key-part:
\be
|| \rho - \gamma ||_1 \geq z(d)
\ee
with
\be
z(d)=\inf_{1\geq \epsilon\geq 0} \left\lbrace \epsilon : \epsilon \geq {1\over 6} - \frac{2 \eta(\epsilon)}{3\log d}\right\rbrace ,
\ee
where $\eta(x)=-x\log x$.
\label{thm:cbound}
\end{theorem}

Having above theorem, let us mention an easy corollary.

\begin{corollary}  
For any private state $\gamma_d$, there is:
\be
\inf_{\Lambda \in LOCC}\inf_{\rho \in \zcal} ||\Lambda({\rho}) - \gamma_d||_1 \geq \frac{1}{6} - \frac{2}{3 \log d}.
\ee
The state $\Lambda(\rho)$ and $\gamma_d$ in the above are assumed to be properly embedded so that their dimensions are compatible.
\label{cor:num4}
\end{corollary}
For such a choice of embedding as described in above proof any operation applied initially to states $\rho_{d_1},\gamma_d$ has to be extended by the identity operator $\mathbf{1}$ acting on additional subsystems. Further we always understand such extensions as a part of the shield in the case of private states.

With a bit of self criticism let us note that the bound in Theorem~\ref{thm:cbound} is actually small, reaching asymptotically $\lim_{d\rightarrow \infty} z(d) = 1/6 =0.1(6)$  as $1/\log d$ instead of $1/d$, although perhaps our considerations can be improved by more careful approach. The bound however is non-trivial even for the smallest $d=2$, with the value $0.041$. It is worth to mention here that proof of theorem~\ref{thm:cbound} can be recalculated also for embedded private state $\gamma'\equiv \gamma_d \ot |0\>\<0|_{d_1}^{\ot 2}$ as it is described in the proof of corollary~\ref{cor:num4} and below it. Namely, the only thing we  have to do is treat additional subsystems $|0\>\<0|^{\ot 2}_{d_1}$ as a part of the shield. Then partial trace in expression~\ref{eq:monot} is well defined as well as untwisting operation $U$.
	
After completing this paper authors where pointed to paper~\cite{Winter2016a}, where estimation tighter than in~\eqref{eq:s} is presented. Unfortunately, in the limit case  both approaches lead to the same value of the lower bound.

\section{Approximate irreducible private states}
\label{sec:approx}
We note that there is a natural way to define approximate irreducible private states as follows (compare with \cite{Ferrara}, where approximate strictly reducible private states were defined):
\begin{definition}
	The state $\rho$ is called $\epsilon$-approximate irreducible private state if it satisfies $||\rho - \tilde{\gamma}_{d,d'}||_1 \leq \epsilon$ for some irreducible private state $\tilde{\gamma}_{d,d'}$.
\end{definition}

A natural property that we would expect from the approximate irreducible private states is that they have the distillable key close to $\log d$. We prove
a weaker condition, namely only for states that approximate those irreducible private states which are from the set $\mathcal{R}_{ir,M}$ with $M$ being asymptotically continuous and bound on the distillable key. The bound then applies for $E_r, E_{r}^{\infty}, I_{sq}$ as they are all asymptotically continuous, and upper bound $K_D$. To be precise, in what follows
by asymptotic continuity of $M$ we mean the fact that for any two states satisfying $||\rho - \rho'||= \epsilon\leq 1$ there is:
\be
|M(\rho) -M(\rho')| \leq O(\epsilon \log d  + h(\epsilon))
\ee
where $h$ is the binary Shannon entropy, and the constant in $O(\cdot)$ notation is independent of the dimension $d$.
\begin{observation}
	\label{obs2a}
	For an $\epsilon$-approximate irreducible private state $\rho$ with $\epsilon \leq 0.367879$ there is:
\be K_D(\rho) \geq \log d - 6\epsilon\log d - 4\eta(\epsilon), \ee
where $\eta(x) = - x \log x$.
If it also satisfies $||\rho - \tilde{\gamma}||_1\leq \epsilon$ such that $\tilde{\gamma} \in \mathcal{R}_{ir,M}$ where $M$ is asymptotically continuous and $M\geq K_D$, then there is:
\be K_D(\rho)\leq \log d + O(\epsilon \log d + h(\epsilon)).
\ee
In particular it is true for $M \in \{E_r,E_r^{\infty},I_{sq}\}$.
\end{observation}

Let us note here, that if we are not interested in optimal constants in the formulation of the asymptotic continuity, then the 
above observation can be stated in much compressed way. Namely that for $||\rho -\gamma_d ||\leq \epsilon \leq {1\over 2}$ there is:
\be
| K_D(\rho) - \log d| \leq O(\epsilon \log d + h(\epsilon)),
\ee
if only $\gamma_d \in {\cal R}_{ir,M}$ with $M \geq K_D$ and $M$ asymptotic continuous.

Finally, it is interesting to note, that there are private states $\gamma_d$, for which it is not known if they are $\epsilon$-approximate irreducible private states, but they possess the major property of the latter, namely $|K_D(\gamma_d)-\log d|\leq \epsilon$. We leave as an open question if the states with that property are indeed the $\epsilon$-irreducible private states, and invoke now the construction of a private state which has this property.
To begin with, we consider a family of states on \newline
$\CC^2\ot\CC^2\ot(\CC^{\widetilde{d}^k}\ot\CC^{\widetilde{d}^k})^{\ot m}$ given in \cite{pptkey}:
\begin{widetext}
\be
\hat{\rho}_{p,\widetilde{d},k,m} ={1\over N_{m}}\left[\begin{array}{cccc}
[p({\tau_1+\tau_2\over 2})]^{\ot m} &0&0&[p({\tau_1-\tau_2\over 2})]^{\ot m} \\
0& [({1\over 2}-p)\tau_2]^{\ot m}&0&0 \\
0&0&[({1\over 2}-p)\tau_2]^{\ot m}& 0\\
{[p({\tau_1-\tau_2\over 2})]}^{\ot m}&0&0&{[p({\tau_1+\tau_2\over 2})]}^{\ot m}\\
\end{array}
\right],
\label{eq:rec-state-presented}
\ee
\end{widetext}
where $N_m = 2(p^m)+2\left( {1\over 2} -p\right) ^m$, $\tau_1 = \left( {\rho_a + \rho_s\over 2}\right) ^{\ot k}$ and $\tau_2=(\rho_s)^{\ot k}$, while $\rho_s$ and $\rho_a$ are the $\widetilde{d}$-dimensional symmetric and antisymmetric Werner state, respectively.

The state $\hat{\rho}_{p,\widetilde{d},k,m}$ is PPT iff $p \leq \frac13$ and
$\frac{1-p}{p} \geq \left( \frac{\widetilde{d}}{\widetilde{d}-1}\right) ^k$ \cite{pptkey}.
We satisfy this condition by setting $p=\frac13$,
$\widetilde{d} =m^2$ and $k = m$, as then $\left( {\widetilde{d}\over \widetilde{d}-1}\right) ^k < 2$ for $m \geq 2$.
Then we define
\be
  \rho_m := \hat{\rho}_{1/3, m^2,m,m},
\ee
with $m \geq 2$.

Now in the proof of Corollary 23 of \cite{BCHW2015} it is shown (together with construction), that there exists $\gamma$ such that $||\rho_m  -\gamma||_1\leq f(\epsilon)$, where $f(\epsilon) = 2\sqrt{4\sqrt{2\epsilon}+ \eta\left( 2\sqrt{2\epsilon}\right) } + 2\sqrt{2\epsilon}$ where $\epsilon = {2\over 3}\left( 1-\left( 1-{1\over 2^m}\right) ^m \times {1\over 1 +{1\over 2^m}}\right) $.

We will argue now that this $\gamma$ is the $O(\exp(-m))$-almost irreducible private state. For the upper bound
on $K_D(\gamma)$ we note that by Lemma 24 of \cite{BCHW2015} there is $K_D(\gamma)\leq E_f(\gamma)\leq 1 + O(\exp(-m))$, where $E_f$ is the entanglement of formation measure. On the other hand the state is private, hence $K_D(\gamma) \geq 1$ by the definition.

\section{ Discussion}

We have considered distillation of the key from private states with non-trivial shield part. We have provided the first protocol, which distills the key not
only from the key part of the private state but also exploits its shield. The protocol is rather intuitive: in the first step we decouple the key part
and the shield by "sorting" the states on the shield conditionally on the value of the key on key part. In the second one we distill the key from the states on shield, and to the
total output state of the original key part and resulting the state we apply the Devetak-Winter protocol. It is plausible, that the protocol is optimal,
however we have proven only a better bound on its rate which reads:
\be
\log d + \frac{1}{d}E_r^{\infty}(\sigma_0\otimes \cdots \otimes\sigma_{d-1}).
\ee
It would be interesting to amen this question. We also leave the problem open  whether our protocol has rate close to~\eqref{rate-thm}, when run on an
approximate irreducible private state. At first it seems possible, as $C_{DW}$ is asymptotically continuous, yet our protocol uses as subprotocol an optimal one for distillation of state $\sigma_0\ot \cdots \ot\sigma_{d-1}$, which need not yield optimal value on a close by state to the latter.

The presented protocol allowed us to characterize the irreducible private states - those for which distillable key equals a logarithm of the dimension of the key part. These states in matrix form have key-undistillable states on diagonal. Formally, for a private is irreducible if and only if
\be
\gamma = {1\over d}\sum_{ij}  |ii\>\<jj|\ot U_i\sigma U_j^{\dagger}
\ee
with $K_D\left( U_i\sigma U_i^{\dagger}\right) =0$, that is, its conditional states are key-undistillable. In turn the private states with $2\otimes 2$ dimensional key part and shield are {\it only} strictly irreducible.
To further characterize them, we observe, that the latter form strictly larger set than the set based on {\it absolutely separable states} as proposed in \cite{KH-phd}.

A major problem left over, which stays behind our considerations is whether there exist entangled key-undistillable states. We show, that if they do not
exist, then the class of irreducible private states collapses to the class of strictly irreducible ones, having $K_D=E_r$. If not, then we may have other
subsets, with $K_D = M$ where $M$ is some entanglement measure which is an upper bound on $K_D$, like $E_r^{\infty}$ and $I_{sq}$. We show that these sets
have useful property of being closed under tensor product.

Finally our result may lead to the generalization of one of the results of \cite{Ferrara}. It is shown there that the one-way key repeater rate of any
	strictly irreducible states $\gamma$ (in two copies) $R_D^{\rightarrow}(\gamma\ot\gamma)$ is upper bounded by the one-way distillable entanglement $E_D^{\rightarrow}(\gamma\ot\gamma)$ (see definitions in Appendix~\ref{AppC}). The result of Ferrara and Christandl~\cite{Ferrara} in a part bases on the fact~\cite{keyhuge}, that the distance between separable and private states approaches 1 exponentially fast in number of qubits that the states occupy. Since it works for separable states, only strictly irreducible private states were considered there~\cite{Ferrara}. It is an open question if our new bound on distance between key undistillable states and
		private ones may lead to analogous bound on repeated key of any irreducible private state.
Answering this question
would show that irreducible private states are useless for transferring key for long distances unless they have high distillable entanglement.

\vspace{0.2cm}
{\bf Acknowledgments} KH thanks Matthias Christandl and Roberto Ferrara for discussions and Mark Wilde for useful comments. KH and MS would like to thank also Marek Winczewski for fruitful discussion on the connection between strictly irreducible private states and the key repeaters. Authors acknowledge grant Sonata Bis 5 (grant num-
ber: 2015/18/E/ST2/00327) from the National Science
Centre. KH and P\'C acknowledge also the EU grant ERC AdG QOLAPS. AR acknowledges support of National Science Centre, Poland, grant 2014/14/M/ST2/00818 and  National Science Centre, Poland, grant OPUS 9. 2015/17/B/ST2/01945. MS also thanks the Seventh framework programme EU grant RAQUEL No 323970 and National Science Centre, Poland, grant OPUS 9.  2015/17/B/ST2/01945 for the support.

\bibliographystyle{apsrev}
\bibliography{references-intro,rmp15-hugekey,references}

\newpage
\widetext
\appendix
\section{Auxiliary definitions from the main text}
\label{app:cdk}
\subsection{Entanglement measures and distillable key}
Here, we introduce all entanglement measures which are employed in this manuscript.  Additionally, the definition of the distillable key is also given.
\begin{definition}
	\label{Er}
	The relative entropy of entanglement for an arbitrary density operator $\rho $ is defined as
	\be
	E_r(\rho) \equiv\mathop{\inf}\limits_{\omega \in \mathcal{SEP}} D(\rho|\omega),
	\ee
	where the infimum runs over the set of separable states $\mathcal{SEP}$, and $D(\cdot|\cdot)$ denotes relative entropy, i.e. $ D(\rho|\sigma)\equiv \tr\rho \log \rho - \tr \rho \log \sigma$, for an arbitrary density operators $\rho,\sigma$.
\end{definition}

\begin{definition}
	\label{ErInf}
	The regularized relative entropy of entanglement for an arbitrary density operator $\rho$ is defined as
	\be
	E_r^{\infty}(\rho)\equiv \mathop{\lim}\limits_{n \rightarrow \infty} \frac{1}{n}E_r\left(\rho^{\ot n} \right),
	\ee
	where $E_r$ is the relative entropy of entanglement given in (\ref{Er}).
\end{definition}

\begin{definition}
	\label{Isq}
	The squashed entanglement \cite{Squashed} for an arbitrary bipartite sate $\rho_{AB}$ is defined as
	\be
	I_{sq}\left(\rho_{AB} \right)\equiv  \mathop{\inf}\limits_{\rho_{ABE}}\left\lbrace \frac{1}{2}I(A;B|E) \ | \ \rho_{ABE} \ \text{extension of} \ \rho_{AB} \right\rbrace.
	\ee
	The infimum is taken over all extensions of $\rho_{AB}$, i.e. over all density operators $\rho_{ABE}$ with $\rho_{AB}=\tr_E\rho_{ABE}$. By $I(A;B|E)\equiv S(AE)+S(BE)-S(ABE)-S(E)$ we denote the quantum conditional mutual information of $\rho_{ABE}$ \cite{CerfAdami}. $S(A)\equiv S(\rho_A)$ is the von Neumann entropy of the underlying state.
\end{definition}

\begin{definition}
	\label{disstkey}
	The distillable key for an arbitrary bipartite state $\rho_{AB}$ is defined as the rate at which private states can be distilled under bipartite LOCC operations $\Lambda_{A:B}$\cite{pptkey}:
	\be
	K_D\left( \rho_{AB}\right)\equiv \mathop{\lim}\limits_{\epsilon \rightarrow 0} \mathop{\lim}\limits_{n \rightarrow \infty} \mathop{\sup}\limits_{\Lambda_{A:B}} \left\lbrace K \ : \ \left|\left|  \Lambda\left(\rho^{\ot n}\right)-\gamma^{nK}\right|\right|_1   \leq \epsilon   \right\rbrace,
	\ee
	where $||\cdot||_1$ denotes the trace norm.
\end{definition}

\begin{definition}
	For an arbitrary bipartite state $\rho_{AB} $ by distillable entanglement~\cite{Plenio1} we understand the following quantity:
	\be
	E_D(\rho_{AB})\equiv \sup \left\lbrace r: \lim_{n \rightarrow \infty} \left[ \inf_{\Lambda }\left| \left| \Lambda(\rho_{AB}^{\ot n})-\Phi(2^{rn})\right|\right|_1 \right] =0\right\rbrace ,
	\ee
	where $\Lambda$ denotes trace preserving bipartite LOCC map and $\Phi(2^{rn})$ denotes maximally entangled state of Schmidt rank $2^{rn}$. The distillable entanglement of a state $\rho_{AB}$ can be understood as the rate at which maximally entanglement states can be distilled under bipartite LOCC operations.
\end{definition}
\subsection{Distillable classical key}
\label{sec:distkey}
Let us recall here the definition of the distillable classical key, which we adapt from \cite{keyhuge}.
\begin{definition}
	\label{cladisstkey}
	The classical distillable key for any given state
$\rho_{ABE}\in B(\hcal_A\ot \hcal_B\ot\hcal_C)$ is defined as the rate at which private states can be distilled under bipartite LOPC operations (local operations and public communication - the difference between the standard local
operations and classical communication (LOCC) and
LOPC lies in the fact that in the latter we need to remember
that any classical message announced by the involved
parties may be registered by Eve) $\Lambda_{A:BE}'$\cite{keyhuge}:
	\be
	C_D\left( \rho_{AB}\right)\equiv \mathop{\lim}\limits_{\epsilon \rightarrow 0} \mathop{\lim}\limits_{n \rightarrow \infty} \mathop{\sup}\limits_{\Lambda_{A:BE}'} \left\lbrace K \ : \ \left|\left|  \Lambda'\left(\rho^{\ot n}\right)-(\rho)_{ccq}^{nK}\right|\right|_1   \leq \epsilon   \right\rbrace,
	\ee
	where $||\cdot||_1$ denotes the trace norm.
\end{definition}

The above definition works for any input tripartite state $\rho_{ABE}$.
However, in the case
where the total state is pure, we observe that the latter is determined
by the state $\rho_{AB}=\tr_E \rho_{ABE}$ up to unitary transformations
on Eve's side. Since from the very definition $C_D$ does not change
under such  transformations,  the latter freedom is not an issue,
so that we can say the state $\rho_{AB}$  completely determines
the total state. Thus, we  get definition of distillable classical
secure key from {\it bipartite } state $\rho_{AB}$
\begin{definition}
For given bipartite state $\rho_{AB}$ the distillable
classical secure key is given by
\be
C_D(\rho_{AB})\equiv  C_D(\psi_{ABE}),
\ee
where $\psi_{ABE}$ is the purification of $\rho_{AB}$.
\label{def:ckd}
\end{definition}
\subsection{Definition of key one-way swapping rate and one-way distillable entanglement}
\label{AppC}
In this section, for the sake of completeness, we recall the definition of the one-way key swapping rate introduced in~\cite{Ferrara} and one-way distillable entanglement used by us in discussion.
\begin{definition}
	\label{def:key swap}
	For all bipartite states $\rho$ and $\widetilde{\rho}$, we define the one-way key swapping rate achieved with one-way key swapping protocols as:
	\be
	\label{eq:key swap}
	\begin{split}
		R_D^{\rightarrow}\left(\rho,\widetilde{\rho} \right)\equiv \lim_{\substack{\delta \rightarrow 0 \\ \epsilon \rightarrow 0 \\ \widetilde{\epsilon}\rightarrow 0}} \lim_{n\rightarrow \infty} \sup_{\substack{\Lambda_{C\rightarrow A:B}\\ \Gamma_{C\rightarrow A}\\ \widetilde{\Gamma}_{C'\rightarrow B}}} \left\lbrace R \ : \ \tr_C \Lambda\left(  \gamma^{\<nr\>} \ot \gamma^{\<n\widetilde{r}\>}\right)  \approx_{\delta}\gamma_{nR}, \ \Gamma(\rho^{\ot n})\approx_{\epsilon}\gamma^{\<nr\>}, \ \widetilde{\Gamma}(\widetilde{\rho}^{\ot n})\approx_{\widetilde{\epsilon}}\gamma^{\<n\widetilde{r}\>} \right\rbrace,
	\end{split}
	\ee
	where by $\gamma^{\<n\widetilde{r}\>},\gamma^{\<nr\>}$ we denote the strictly irreducible private state with dimension of the key part $\lceil n\widetilde{r} \rceil $ and $\lceil nr \rceil$ respectively, by $\gamma_{nR}$ the private state with dimension of the key part equal to $\lceil nR \rceil$. Moreover, by $\Lambda_{C\rightarrow A:B}, \Gamma_{C\rightarrow A}$ and $\widetilde{\Gamma}_{C'\rightarrow B}$ we denote one-way swapping protocols acting on respective subsystems.
\end{definition}
\begin{definition}
	For all bipartite states $\rho$ we define one-way distillable entanglement
	\be
	E_D^{\rightarrow}\equiv=\lim_{\epsilon \rightarrow 0}\lim_{n\rightarrow \infty}\sup_{\Lambda_{A\rightarrow B}}\left\lbrace E: \Lambda \left(\rho^{\ot n} \right)\approx_{\epsilon}\Phi_{AB}i^{nE}  \right\rbrace, 
	\ee
	where maps $\Lambda_{A\rightarrow B}$ are restricted to one-way $LOCC$ and $\Phi_{AB}$ is maximally entangled state between $A$ and $B$.
\end{definition}
In the above expressions we use the notation $\rho \approx_{\epsilon} \sigma$ for $\left| \left|\rho-\sigma \right| \right|_1\leq \epsilon$ to compress the definitions.

\subsection{Classical-classical-quantum states and privacy squeezing}
\label{sec:ccq}
{\definition For a bipartite state $\rho_{ABA'B'}$ its ccq state with respect to system $AB$ and basis ${\cal B} =\{|e_i\>\ot |f_j\>\}$ is the state
	\be
	(\rho)_{ccq} \equiv \tr_{A'B'}\left(\sum_{ij} P_i\ot Q_j \ot \mathbf{1}_E |\psi_{\rho}\>\<\psi_{\rho}| P_i\ot Q_j\ot \mathbf{1}_E\right),
	\ee
	where $P_i \equiv |e_i\>\<e_i|\ot \mathbf{1}_{A'B'}$ and $Q_j=|f_j\>\<f_j|\ot \mathbf{1}_{A'B'}$, while $|\psi_{\rho}\>$ is a purification of the state $\rho_{AB}$ to system $E$. In case when the system $A'B'$ is absent, we say that $(\rho)_{ccq}$ is a ccq state of $\rho$ with respect to basis $\cal B$.
}

Let us note here, that such defined ccq state is not unique, however any two ccq states of $\rho$ can differ only by an isometry
on system $E$, hence the secure content of them is the same, as such operation is at hand of Eve, and is reversible.  In particular we may fix for the rest of our considerations that we use the standard purification:
\be
\sum_i p_i|\psi_i\>\<\psi_i|_{AB} \mapsto \sum_i \sqrt{p_i} |\psi_i\>_{AB}\ot|i\>_E.
\ee

An important fact for what follows is  that the twisting with control basis ${\cal B}=\{|e_i\>\ot|f_i\>\}$  (that is unitary $U=\sum_{ij}|e_i\>\<e_i|\ot|f_j\>\<f_j|_{AB}\ot U^{(e_i,f_j)}_{A'B'}$)  does not change
the ccq state with respect to system $AB$ and  basis $\cal B$ of a given state $\rho_{ABA'B'}$ (\cite{keyhuge}, see also Theorem 3.3 of \cite{KH-phd}). Formally
we have:
\be
(U\rho_{ABA'B'} U^{\dagger})_{ccq} = (\rho_{ABA'B'})_{ccq}.
\ee

We note now that the rate of the one-way Devetak-Winter  protocol on some ccq state $\rho$ which we invoke in Theorem $1$ reads:
\be
C_{DW}(\rho_{ccq}) \equiv I(A:B)_{\rho_{ccq}} - I(A:E)_{\rho_{ccq}},
\label{eq:cdw}
\ee
where $I(X:Y)_{\rho}$ is the (quantum) mutual information of a state of $X,Y$ subsystem of the state $\rho_{ccq}$. For the more formal definition of $C_{DW}$ (the classical distillable key) we refer to Section \ref{sec:distkey} of this appendix.

Finally, let us recall the notion of privacy squeezed state \cite{keyhuge}. It involves the operation of privacy squeezing. Privacy squeezing of $\rho$ according to basis $\{|e_i\>\ot|f_j\>\}$ is composition of  (i) operation of the form $U=\sum_{ij}|e_i\>\<e_i|_{AB}\ot|f_j\>\<f_j|\ot U^{(e_i,f_j)}_{A'B'}$ with certain special $U^{(e_i,f_j)}$ defined by the state $\rho$ and subsequently (ii) tracing out the subsystem $A'B'$. The result of this operation is called the {\it privacy squeezed state} of $\rho$.
\section{Proofs from Section~\ref{sec:protocol}}
\label{app:A}
First, we are going to prove here the main theorem of the manuscript, namely that we have a protocol that can distill key not only from the key part
but also from the shield of a private state in case where $\sigma_0\ot \cdots \ot\sigma_{d-1}$ is distillable.

Let us recall that it key rate is given the following theorem:
\begin{Th}
For a private state $\gamma_d$, there is
\ben
K_D(\gamma) \geq \log d + \frac{1}{d}K_D(\sigma_0 \ot \sigma_1 \ot \cdots \ot \sigma_{d-1}).
\een
\end{Th}

Recall that the input to the protocol $P(m,k)$ is $\gamma_d^{\ot (m\, k)}$, where $\gamma$ is a private state with $d\ot d$ dimensional key part. The private state $\gamma_d$ has also finite dimensional shield, which without loss of generality we may assume to be of dimension $d'\ot d'$. Parameters $m$ and $k$ will be chosen later. We first recall the protocol:

\begin{enumerate}
\item For each of $k$ blocks of $m$ states Alice and Bob repeat items $2$-$6$ as follows:
\item Alice and Bob measure their subsystems of the key part $A$ and $B$.
\item If number $t_s(m)$ of any symbol $s\in \{0,...,d-1\}$ in the key $|i_0\>\ot \cdots \ot|i_{m}\>$ with $i_j\in\{0,...,d-1\}$ is below the threshold $\delta(m)$, i.e. $|{m\over d} - t_s(m)|>\delta(m)$, they trace out the original state, produce the error state $|e\>_{AA'}\ot|e\>_{BB'}$. They then proceed with step $2$ on the next block if such is left, or go directly to step $7$ otherwise. If however $|{m\over d} -t_s(m)|\leq  \delta(m)$ for all $s$, they proceed with the steps $4-6$ of the protocol.
\label{item:ab1}
\item Alice and Bob perform each control-permutation operation: depending on the value of the key $|i_0\>\ot\cdots\ot|i_{m}\>$ with $i_j\in\{0,\ldots,d-1\}$
they permute states on the systems $A'B'$ from its original form $\sigma_{i_0}\ot \cdots\ot\sigma_{i_{m}}$ to $\sigma_0^{\ot t_0(m)}\ot \cdots\ot\sigma_{d-1}^{\ot  t_{d-1}(m)}$.
\item If necessary Alice and Bob trace out some of the states $\sigma_0,\ldots,\sigma_d$ leaving only $t_{min}(m)=\lfloor \frac{m}{d} - \delta(m) \rfloor$ of copies of the state $\widetilde{\sigma} \equiv \sigma_0\ot \cdots \ot\sigma_{d-1}$.
\label{item:trunc}
\item As a subprotocol, Alice and Bob perform on systems $A'B'$ any optimal key distillation protocol of the state $\sigma$, acting on $(\sigma_0\ot \cdots \ot \sigma_{d-1})^{\ot t_{min}(m)}$, with output close in the trace norm
to the state $\gamma_{A'B'A''B''}$. The state on systems $AB$ from the item \ref{item:ab1} and state on systems $A'B'A''B''$ form one of the $k$ input states for the final step $7$ of the protocol.
\label{item:best-dist}
\item On $k$ output states of the $k$ repetitions of items $2$-$6$ (or items $2-3$ if unsuccessful in item $3$ respectively) Alice and Bob perform the Devetak-Winter protocol \cite{DevetakWinter-hash}.

\end{enumerate}

As it was argued in the main text, in what follows, we can w.l.g. assume $K_D(\sigma_0\ot \cdots \ot \sigma_{d-1})>0$, as otherwise the bound from theorem is trivially satisfied without need for the above protocol.

We will now examine how the output state of the protocol looks like. Let us fix some $\delta'(m) > 0$. As we already noticed, we start from a pbit $\gamma^{\ot m}_{ABA'B'E}$. In what follows we will write $|\psi\>$ instead of $|\psi\>\<\psi|$ to simplify the notation. After projection on the key part and control-sorting of the key part and shield, the state is:
\be
\label{eq:state-in}
\begin{split}
\sum_{t \in G} {|Q_t|\over d^m}\left[ \sum_{i \in Q_t}{1\over|Q_t|} |ii\>\<ii|_{AB}\right] \otimes |\psi_{0}\>^{\otimes t_0(m)}\otimes \cdots \otimes |\psi_{d-1}\>^{\otimes t_{d-1}(m)}
+ p_B |e\>_{AA'} \otimes |e\>_{BB'} \otimes |e\>_E.
\end{split}
\ee
Here $|\psi_0\>$ is purification of $\sigma_0 = U_0\sigma_{A'B'}U^{\dagger}_{0}$, and similar for $|\psi_k\>$ for $k \in \{0,\ldots,d-1\}$. The multi-index $i =(i_1,\ldots,i_m)$, where $i_j \in \{0,\ldots,d-1\}$, the set $G$ is the set of types satisfying
$|t_s  - \frac{1}{d}|\leq \delta(m)$ for all $s \in \{0,\ldots,d-1\}$, and the set of all strings $i$ of the type $t$ is denoted as $Q_t$. To simply the notation, we will write Eq. (\ref{eq:state-in}) as
\be
\sum_{t \in G} p(t) \rho_t \otimes |\psi_{0}\>^{\otimes t_0(m)}\otimes \cdots \otimes |\psi_{d-1}\>^{\otimes t_{d-1}(m)}
+ p_B |e\>_{AA'} \otimes |e\>_{BB'} \otimes |e\>_E,
\ee
where $p(t)=\frac{|Q_t|}{d^m}$ and $\rho_t=\sum_{i \in Q_t}{1\over|Q_t|} |ii\>\<ii|_{AB}$.
Moreover we have
\be
p_B = 1 -  \sum_{t\in G} p(t).
\ee
The state $|e\>_{AA'} \otimes |e\>_{BB'} \otimes |e\>_{E'}$ is an error state produced when there are
too little of $\sigma_0 \ldots \sigma_{d-1}$. From properties of types, there is $p_B \leq p(m)$ \cite{Renner-Phd} with
\be
p(m)= 2^{-m\left( \frac{\delta(m)^2}{2 \ln 2} - d\frac{\log (m+1)}{m}\right) },
\ee
i.e. decays exponentially fast with $m$. Moreover the cardinality  of the set $\cal T$ of all types reads the following value and a bound:
\be
|{\cal T}| = {m+d-1\choose m} \leq 2^{(m+d-1)h\left( \frac{m}{m+d-1}\right) },
\label{eq:nroftypes}
\ee
where $h$ is the Shannon binary entropy.

According to item \ref{item:trunc} of the protocol Alice and Bob trace (if needed) some of states $\sigma_i$, leaving $t_{min}(m) = \lfloor \frac{m}{d} - \delta(m)\rfloor$ of states $\sigma_0\ot \cdots\ot \sigma_{d-1}$. For simplicity we will denote $t_{min}(m)$ as $t(m)$. Without loss of generality we may assume that Eve knows the very type of string $i$ of system $AB$ which we denote as $|t\>\<t|_{\tilde{E}}$. Indeed, she may actually conclude $t$ from the kind and number of states traced out in item \ref{item:trunc}.
Altogether the input state in item \ref{item:best-dist} of the protocol is of the form :
\be
\rho_{in} = \sum_{t \in G} p(t) \rho_t \otimes (|\psi_{0}\>\otimes \ldots \otimes |\psi_{d-1}\>)_{A'B'E}^{\otimes t(m)}\otimes |t\>\<t|_{\tilde{E}}
+ p_B |e\>\<e|_{AA'} \otimes |e\>\<e|_{BB'} \otimes |e\>\<e|_{E\tilde{E}},
\label{eq:state-rhoin}
\ee
where $|\psi_k\> \equiv |\psi_k\>\<\psi_k|$.

On the system $A'B'$ Alice and Bob have then $A'B'$ subsystem of the state $(|\psi_0\>\ot|\psi_1\>\ot \cdots\ot |\psi_{d-1}\>)^{\ot t(m)}$. Now from the very definition of distillable key, for any good $t$ (that is $t \in G$), from this they distill by some operation $\Lambda_{t(m)}$ a state
\be
||\hat{\gamma}_{d_{t(m)}} - \gamma_{d_{t(m)}} ||_1\leq \epsilon_{t(m)}'
\label{eq:close}
\ee
on systems ${A'B'A''B'}$
such that
\be
{\log d_{t(m)} \over t(m)} \geq K_D(\sigma_0\ot \cdots \ot\sigma_{d-1}) - \delta'(m)
\ee
$\gamma_{A'B'A''B'}$ is a private state. For bad $t$ (that is $t \notin G$), $\Lambda_{t(m)}$ creates additional flags $|e\>_{A''}\ot |e\>_{B''}$. Without loss of generality we
may assume that it is accompanied by creation of state $|e_{E''}\>$.
To analyse the key rate of the considered protocol it will be convenient to use the notion of classical distillable key \cite{DevetakWinter-hash} $C_D$ (for the definition, see Appendix). It is shown
that if there exists protocol which distills from $\rho$ a state close to private state by $\epsilon$ via LOCC operation, then there also exists a protocol with {\it local operations and public communication}, which works on purification of the input state $\rho$ and distills state that is close to an {\it ideal ccq state} by $2\sqrt{\epsilon}$ (see theorem 4.11 of \cite{KH-phd}), with the same rate of key ($K_D(\rho)=C_D(\psi_{\rho})$). In our case $K_D\left( (\sigma_0\ot \cdots \ot \sigma_{d-1})^{\ot t(m)}\right) =C_D\left( (|\psi_0\>\ot \cdots \ot |\psi_{d-1}\>)^{\ot t(m)}\right) $. Hence, there exists LOPC operation $\tilde{\Lambda}_{t(m)}$ which distills state $\hat{\alpha}_{A'B'E'}(d_{t(m)})$ of the ccq form $\sum_{i,j=0}^{d_{t(m)}}|ij\>\<ij|_{A'B'}\ot \rho^{(ij)}_{E'}$, such that:
\be
||\hat{\alpha}_{A'B'E'}(d_{t(m)}) - \alpha_{ideal}(d_{t(m)})||_1\leq 2\sqrt{\epsilon'_{t(m)}},
\label{eq:close-alf}
\ee
where
\be
\alpha_{ideal}(d_{t(m)}) = {1\over d_{t(m)}} \sum_{k=0}^{d_{t(m)-1}} |kk\>\<kk|_{A'B'}\ot \rho_{E'}.
\ee

As we will see, to lower bound the key rate of the considered protocol, it is enough to analize the state which would be the output of $\tilde{\Lambda}_{t(m)}$:
\be
\begin{split}
\rho_m = \sum_{t \in G} p(t) \rho_t \otimes \hat{\alpha}_{A'B'E'}\left( d,t(m)\right) \ot |t\>\<t|_{\tilde{E}}
+ p_B |e\> \<e|_{AA'}\otimes |e\> \<e|_{BB'} \otimes |e\> \<e|_{E''},
\end{split}
\label{eq:rhom}
\ee
where $|t\>\<t|$ is the state which describes Eve's knowledge about type $t$ of her total state on $E'' = E'\tilde{E}$.
Our figure of merit will be state which is close by in the trace norm distance, and simpler to analyse. This is the above state with $\alpha_{ideal}$ in place of $\hat{\alpha}(d_{t(m)})$:

\ben
\label{rho-prime}
\rho_m' =\sum_{t \in G} p(t) \rho_t \otimes (\alpha_{ideal})_{A'B'E'}\ot|t\>\<t|_{\tilde{E}}
+ p_B |e\> \<e|_{AA'} \otimes |e\> \<e|_{BB'}  \ot |e\> \<e|_{E''}.
\een
To see that the above state is close to $\rho_m$ observe that for all $t,t' \in G$ such that $t' \neq t$ there is $\rho_t\ot|t\>\<t| \perp \rho_{t'}\ot |t'\>\<t'|$. We can also assume w.l.g. that symbol $e$ of the error states $|e\>_{AA'}\ot|e\>_{BB'}\ot|e\>_{E''}$  does not belong to the set $\Sigma=\{0,\ldots,d-1\}^{\times m}$. Thanks to these two facts, via properties of the trace norm and  Eq. (\ref{eq:close-alf}) we obtain:
\be
||\rho_m - \rho_m'||_1 < \epsilon_{t(m)},
\label{eq-close-rho}
\ee
where $\epsilon_{t(m)} \equiv 2\sqrt{\epsilon'_{t(m)}}$.

We will now reduce the problem of calculating the rate of the considered protocol to study of the Devetak-Winter formula of certain state as follows, which
shows that it is enough to analyse the state (\ref{eq:rhom})

\ben
\begin{split}
\forall_{m,k}\quad  mkK_D(\gamma_d) \geq  K_D(\gamma_d^{\ot m\, k})=  C_D(\psi_{\gamma_d}^{\ot m\, k})
\geq C_D(\rho_{in}^{\ot k}) \geq C_D(\rho_m^{\ot k})\geq  {\cal R}_{DW}(\rho_m^{\ot k}),
\end{split}
\een
In first inequality we use definition of $K_D$ which, as operational entanglement measure is regularized by definition. We then use the aforementioned equivalence
between $C_D$ and $K_D$ where $C_D$ is calculated on the purification of input state $|\psi_{\gamma_d}\>$ to system of Eve. We then use the fact that $C_D$ is non-increasing under operations of Alice and Bob. By ${\cal R}_{DW}$ we denote the rate of the (one-way) Devetak-Winter protocol run for $k$ copies of the state. In the last step we follow \cite{Christandl-Ekert-etal}, observing that it need not be optimal, yet is good enough for our purpose, hence the last inequality follows.
From \cite{DevetakWinter-hash} we have that:
\be
\label{cdw-rho}
{\cal R}_{DW}(\rho_m^{\ot k}) = k\left( I(AA':BB')_{\rho_m} - I(AA':E'')_{\rho_m}\right)  - o(k) \equiv kC_{DW}(\rho_m) - o(k),
\ee
hence, in the limit of large $k$ there is:
\be
K_D(\gamma_d) \geq {1\over m}\left( I(AA':BB')_{\rho_m} - I(AA':E'')_{\rho_m}\right) \equiv {1\over m} C_{DW}(\rho_m),
\ee
for all $m$. It is then enough to find lower bound on the RHS of the above inequality. We will find it using asymptotic continuity of entropy (and conditional entropy), based on the fact that $\rho_m$ is close to $\rho_m'$ for which it is rather easy to calculate the entropies. Lemma \ref{lemma-cdw} states that
\be
C_{DW}(\rho_m)\geq m\log d + (1- p_B)t(m)(K_D(\sigma)-\delta(m)) - (m+d-1) h\left( \frac{m}{m+d-1}\right)  - f(\epsilon_{t(m)},d,d_{t(m)},m)
\ee
with $\lim_{m\rightarrow \infty,\epsilon_n \rightarrow 0} {1\over m}f(\epsilon_{t(m)},d,d_{t(m)},m) = 0$. Thus in limit of large $m = s d$  ($ s \rightarrow \infty$) we obtain the desired lower bound. To see this observe that we can set  $\delta(m) = \frac{2 \log m \sqrt{d}}{\sqrt{m}}$ which assures $p_B\rightarrow 0$. Further, $\frac{t(m)}{m} \geq \frac{1}{d} - \delta(m)$, thus we have:
\be
t(m)\left( K_D(\widetilde{\sigma}) - \delta'(m)\right)  \geq \frac{1}{d} K_D(\widetilde{\sigma}) - \delta(m)K_D(\widetilde{\sigma}) -\delta'(m).
\ee
Finally $K_D(\widetilde{\sigma})\leq d\log d'$ where $d'$ is the dimension of the shield part of $\gamma_d$, hence the pre-last term in the above tends to $0$ with $m$ for our choice of $\delta(m)$. Now, since $h(1)=0$, the term $h(\frac{m}{m+d-1})$ vanishes asymptotically as well. Finally $\delta'(m)$ was arbitrarily small, hence the thesis. \qed




\begin{lemma}
\label{lemma-cdw}
For the state $\rho_m$ from Eq. \eqref{eq:rhom}, there is:
\be
\begin{split}
C_{DW}(\rho_m) \geq \,&m\log d + (1- p_B)t(m)\left( K_D(\sigma_0 \ot \cdots \ot \sigma_{d-1})-\delta(m)\right)\\
  &- (m+d-1)h\left( \frac{m}{m+d-1}\right)  - f(\epsilon_{t(m)},d,d_{t(m)},m)
\end{split}
\ee
with $\lim_{m\rightarrow \infty,\epsilon_n \rightarrow 0} {1\over m}f(\epsilon_{t(m)},d,d_{t(m)},m) = 0 $.
\end{lemma}

Inserting the state $\rho_m$ into Eq. \eqref{cdw-rho} gives
\be
C_{DW}(\rho_m) = I(AA':BB')_{\rho_m} - I(AA':E'')_{\rho_m}.
\ee
However, we will compute the above quantity for $\rho_m'$ and will use the asymptotic continuity of $C_{DW}$.

We will first lower bound $I(AA':BB')$. From the definition of quantum mutual entropy, it is equal to
\be
I(AA':BB') = S(AA') - S(AA'|BB'),
\ee
where $S(AA'|BB')$ is the quantum conditional entropy. We will also use the Fannes inequality for the asymptotic continuity of the von Neumann entropy \cite{Fannes1973}, i.e. for $||\rho_1 - \rho_2||_1 \leq \epsilon$ we have
\be
\label{Fannes-cont}
|S(\rho_1) - S(\rho_2)| \leq ||\rho_1 - \rho_2||_1 \log d + \eta(||\rho_1 - \rho_2||_1),
\ee
where $\eta(x) = -x \log x$, and the Alicki-Fannes inequality for continuity of the conditional entropy \cite{Alicki-Fannes-04}:
\be
\label{Fannes-Alicki-cont}
|S_{X|Y}(\rho_1) - S_{X|Y}(\rho_2)| \leq 4||\rho_1 - \rho_2||_1 \log d_X + h\left( ||\rho_1 - \rho_2 ||_1\right),
\ee
where here $X,Y$ denotes two subsystems of $\rho_1$ and $\rho_2$.

Direct calculations of entropies for the state $\rho_m'$ give for each part
\be
S(AA') = (1 - p_B) \log d_{t(m)} + S\left( \sum_{T} p(T) \tilde{\rho}_T\right) ,
\label{eq:entrAA}
\ee
where for $T =\left\{t: t \in G \vee t = t_B \right\}$, $p(T) = p(t)$ for $t \in G$ and $p(T = t_B) = p_B$. Then, for $T \in G$, $\rho_T = \rho_t = \tr_B \sum_{i \in Q_t}{1\over|Q_t|} |ii\>\<ii|_{AB}$
and for $T = t_B$, $\tilde{\rho}_T = |e\>\<e|_{A}$.

We know that
\be
\log d_{t(m)} \geq t(m)\left( K_D(\sigma) - \delta(m)\right) .
\label{eq:tm-delta}
\ee
The entropy of $BB'$ is the same as $S(AA')$. It remains to calculate the joint entropy $S(AA'BB')$, which is equal to $S(AA')$, since the systems $AA'$ and $BB'$
are maximally correlated.
We are going to lower bound $I(AA':BB')_{\rho_m}$ as follows:
\be
\left| I(AA':BB')_{\rho_m} - I(AA':BB')_{\rho_m'}\right| \leq \left| S(AA')_{\rho_m} -S(AA')_{\rho_m'}\right|   + \left| S(AA'|BB')_{\rho_m} -S(AA'|BB')_{\rho_m'}\right| .
\ee
Now in the RHS we apply the asymptotic continuity of the entropies obtaining :
\be
\left| I(AA':BB')_{\rho_m} - I(AA':BB')_{\rho_m'}\right| \leq 5\epsilon_{t(m)} \log \left( d^m\,d_{t(m)}\right)  + 3 h\left( \epsilon_{t(m)}\right) .
\ee
hence for $g_1'\left( \epsilon_{t(m)},d,d_{t(m)},m\right)  = 5\epsilon_{t(m)} \log \left( d^m\,d_{t(m)}\right)  + 3 h\left( \epsilon_{t(m)}\right) $ we have:
\be
I(AA':BB')_{\rho_m} \geq (1 - p_B) \log d_{t(m)} + S\left( \sum_{T} p(T) \tilde{\rho}_T\right)  - g_1\left( \epsilon_{t(m)},d,d',m\right).
\ee
Finally we observe that
\be
\left| \left| \sum_{T} p(T) \tilde{\rho}_T - \frac{\mathbf{1}}{d^m}\right| \right|_1 \leq 2 p_B.
\ee
thus, again via the Fannes inequality:
\be
I(AA':BB')_{\rho_m} \geq (1 - p_B) \log d_{t(m)} + m\log d - 2p_B m\log d - h(2 p_B) - g_1'(\epsilon_{t(m)},d,d_{t(m)},m)
\label{eq:ab}
\ee
denote
\be
g_1\left( \epsilon_{t(m)},d,d_{t(m)},m\right)  \equiv 2p_B m\log d + h(2 p_B) + g_1'\left( \epsilon_{t(m)},d,d_{t(m)},m\right).
\ee


We are left with calculating $I(AA':E'')$. It is equal to
$ I(AA':E'') = S(AA') - S(AA'|E'')$, where $S(AA'|E'') = S(AA'E'') - S(E'')$ is the conditional entropy.
The entropy $S(AA')$ is already bounded in (\ref{eq:entrAA})
Next, we have
\be
S(E'') = (1 - p_B) S(\rho_{E'}) + H\left( \left\{ p(T) \right\}\right) ,
\ee
where $\rho_{E'} = \tr_{A'B'} \alpha_{ideal}$ and $H(\cdot)$ is the Shannon entropy. Finally,
\be
S(AA'E'') = H\left( \left\{ p(T) \right\}\right)  + (1 - p_B)\left( \log d_{t(m)} + S(\rho_{E'})\right)  + \sum_{t \in G} p(t) S(\rho_t).
\ee

Summing up, we have
\ben
\begin{split}
I(AA':E'')_{\rho_m'}& \leq - H\left( \left\{ p(T) \right\}\right)  + H\left( \left\{ p(T) \right\}\right)  + S\left( \sum_{T} p(T) \tilde{\rho}_T\right)  - \sum_{t \in G} p(t) S(\rho_t)\\
&= S\left( \sum_{T} p(T) \tilde{\rho}_T\right)  - \sum_{T} p(T) S(\tilde{\rho}_T) =: \chi,
\end{split}
\een

By the asymptotic continuity of entropy and conditional entropy, we have:
\be
\left| I(AA':E'')_{\rho_m} - I(AA':E')_{\rho_m'}\right|  \leq \epsilon_{t(m)} \log d_{t(m)} + \eta(\epsilon_{t(m)}) + 4\epsilon_{t(m)} \log \left( d^m\,d_{t(m)}\right)  + 2h\left( \epsilon_{t(m)}\right) ,
\ee
but, since $I(AA':E')_{\rho_m'} \leq \chi$,
\be
I(AA':E'')_{\rho_m} \leq \chi + \epsilon_{t(m)} \log d_{t(m)} + \eta(\epsilon_{t(m)}) + 4\epsilon_{t(m)} \log (d^m\,d_{t(m)}) + 2h(\epsilon_{t(m)} ).
\label{eq:ae}
\ee
Let us set
\be
g_2(\epsilon_{t(m)},d,d',m) \equiv \epsilon_{t(m)} \log d_{t(m)} + \eta(\epsilon_{t(m)}) + 4\epsilon_{t(m)} \log (d^m\,d_{t(m)}) + 2h(\epsilon_{t(m)} ).
\ee
Knowing that the Holevo quantity is upper bounded by the entropy of the signal distribution we have \cite{Nielsen-Chuang}
\be
\chi =S\left( \sum_{T} p(T) \tilde{\rho}_T\right)  - \sum_{T} p(T) S(\tilde{\rho}_T) \leq H(\left\{p(T)\right\}) \leq H(\{p(t)\}).
\ee
Now entropy of the "distribution of types" is less than the $\log $ of the support of this distribution which is the number of types.
Hence by (\ref{eq:nroftypes}) we have:
\be
S\left( \sum_{T} p(T) \tilde{\rho}_T\right)  - \sum_{T} p(T) S(\tilde{\rho}_T) \leq (m+d-1) h\left( \frac{m}{m+d-1}\right) .
\ee
and thus $I(AA':E'') \leq (m+d-1) h\left( \frac{m}{m+d-1}\right)  + g_2\left( \epsilon_{t(m)},d,d_{t(m)},m\right) $.

Summing up bounds (\ref{eq:ab}) and (\ref{eq:ae}) we obtain:
\be
\begin{split}
C_{DW}(\rho_m) \geq  (&1 - p_B) \log d_{t(m)} + m\log d -g_1\left( \epsilon_{t(m)},d,d_{t(m)},m\right) \\
&- (m+d-1) h\left( \frac{m}{m+d-1}\right)  - g_2\left( \epsilon_{t(m)},d,d_{t(m)},m\right) .
\end{split}
\ee
Due to (\ref{eq:tm-delta}) there is:
\be
\begin{split}
C_{DW}(\rho_m) \geq  (&1 - p_B) t(m)\left( K_D(\widetilde{\sigma})  - \delta(m)\right) + m\log d -g_1\left( \epsilon_{t(m)},d,d_{t(m)},m\right) \\
&- (m+d-1) h\left( \frac{m}{m+d-1}\right)  - g_2\left( \epsilon_{t(m)},d,d_{t(m)},m\right).
\end{split}
\ee
Define now $f= g_1 + g_2$ and observe, that for $m = s d$ with $s \rightarrow \infty$ such defined $f\left( \epsilon_{t(m)},d,d_{t(m)},m\right) $ divided by $m$ approaches 0 with $\epsilon_{t(m)}$ approaching $0$. Indeed apart from bounded terms in $f$ such as $h\left( \epsilon_{t(m)}\right) $ which are clearly sublinear in $m$, we have terms
$O\left( \epsilon_{t(m)} \log \left( d^m d_{t(m)}\right) \right) $. Such term breaks into two $\epsilon_{t(m)}m \log d$ and $\epsilon_{t(m)}\log d_{t(m)}$. The first term is sublinear in $m$
if only $\epsilon_{t(m)} \rightarrow 0$. This is assured by $ m = s d$ so that $t(m) \approx \frac{m}{d} \rightarrow \infty$. The second term is equal to
\be
\begin{split}
\epsilon_{t(m)} \log {\left( (d')^{d}\right) }^{t(m)} =\epsilon_{t(m)}t(m) d \log d'\leq \epsilon_{t(m)} \left(  {m\over d} + \delta(m)\right) d\log d' =  \epsilon_{t(m)} m \log d' + \epsilon_{t(m)} \delta(m) d \log d'.
\end{split}
\ee
The first term of the RHS of last equality is sublinear in $m$ as we argued above, while the second vanishes for our choice of
$\delta(m) = \frac{2 \log m \sqrt{d}}{\sqrt{m}}$. Hence the assertion follows.
\qed

\begin{proof}[ Proof of Theorem~\ref{thm:upper-bound-er}]
	The first part of the proof is based on Theorems 3 and 4 from \cite{keyhuge}. For the sake of completeness we present here the crucial steps which are necessary to prove our result. From Theorem 3 of the above-mentioned paper for the state $\gamma_{ABA'B'}^{\ot n}$ we have $E_r\left( \gamma_{ABA'B'}^{\ot n}\right) \leq \operatorname{log}d^n+\frac{1}{d^n}\sum_{k=0}^{d^n-1}E_r(\rho_k)$, where $k=(i_1,\ldots,i_n)$ is a multi-index with $i_l\in \left\lbrace 0,1,\ldots,d-1 \right\rbrace $ for every $l\in \left\lbrace 1,\ldots,n\right\rbrace $, and $\rho_k=\sigma_{i_1}\ot \cdots \ot \sigma_{i_n}$. This equation divided by $n$ from the both sides can be easily rewritten as $\frac{1}{n}E_r\left( \gamma_{ABA'B'}^{\ot n}\right)\leq \operatorname{log}d+\frac{1}{nd^n}\sum_{k=0}^{d^n-1}E_r(\rho_k)$, where the left-hand-side for $n\rightarrow \infty$ approaches $E_r^{\infty}\left(\gamma_{ABA'B'}^{\ot n}\right) $.  First of all, let us recall that $E_r(\rho_k)=E_r(\rho_{k'})$ which has the same numbers of occurrence of symbols from the set $\left\lbrace 0,\ldots,d-1 \right\rbrace $, i.e. when $k=k'$. Additionally, considering only those $\rho_k$ for which $k$ is $\delta-$strongly typical ($\delta>0$) we have
	\be
	\label{typicality}
	\forall a\in \left\lbrace 0,\ldots,d-1\right\rbrace\quad  \left|\frac{a(k)}{n}-\frac{1}{d}\right|< \delta,
	\ee
	where $a(k)$ denotes frequency of a symbol $a$ in a sequence $k$. The set of such $k$ we denote as $\mathcal{ST}^n_{\delta}$. Rewriting $\sum_k E_r(\rho_k)$ over $k\in \mathcal{ST}_{\delta}^n$ and $k\notin \mathcal{ST}_{\delta}^n$ due to\cite{keyhuge} we have
	\be
	\label{s1}
	\frac{1}{nd^n}\sum_{k=1}^{d^n-1}E_r(\rho_k) \leq \frac{1}{nd^n}\sum_{k\in \mathcal{ST}_{\delta}^n}E_r(\rho_k)+\epsilon \operatorname{log}d,
	\ee
	where $\epsilon>0$. Now, our goal is to improve the bound on $E_r(\rho_k)=E_r\left(\sigma_0^{\ot m_0}\ot \sigma_1^{\ot m_1}\ot \cdots \ot \sigma_{d-1}^{\ot m_{d-1}} \right)$ when $n\rightarrow \infty$, where $\sigma_i$ are conditional states on $A'B'$.
	Defining $m\equiv \operatorname{min}\left\lbrace m_0,m_1,\ldots,m_{d-1} \right\rbrace $ we can make the following observation using the fact on strong typicality (\ref{typicality})
	\be
	\forall i\in\{0,1,\ldots, d-1\} \ \left|\frac{m_i}{n}-\frac{1}{d} \right|<\delta \ \wedge \ \left|\frac{m}{n}-\frac{1}{d} \right|<\delta.
	\ee
	By the triangle inequality we have
	\be
	\label{5}
	\left|\frac{m}{n}-\frac{m_i}{n} \right|\leq \left|\frac{m}{n}-\frac{1}{d} \right|+\left|\frac{1}{d}-\frac{m_i}{n} \right|<2\delta.
	\ee
	Using the definition of the parameter $m$ and triangle inequality we are able to rewrite the first term on the RHS of (\ref{s1}) as
	\be
	\label{6}
	\begin{split}
		&\frac{1}{nd^n}\sum_{k\in \mathcal{ST}_{\delta}^n}E_r\left(\sigma_0^{\ot m_0}\ot \sigma_1^{\ot m_1}\ot \cdots \ot \sigma_{d-1}^{\ot m_{d-1}} \right) \\ &\leq \sum_{k\in \mathcal{ST}_{\delta}^n}\frac{m}{nd^n}\frac{1}{m}E_r\left(\widetilde{\sigma}^{\ot m} \right)+\sum_{k\in \mathcal{ST}_{\delta}^n}\frac{\widetilde{m}_0}{nd^n}\frac{1}{\widetilde{m}_0}E_r\left(\sigma_0^{\ot \widetilde{m}_0} \right)\\
		&+\cdots+ \sum_{k\in \mathcal{ST}_{\delta}^n}\frac{\widetilde{m}_{d-1}}{nd^n}\frac{1}{\widetilde{m}_{d-1}}E_r\left(\sigma_{d-1}^{\ot \widetilde{m}_{d-1}}\right),
	\end{split}
	\ee
	where $\widetilde{\sigma}=\sigma_0\ot \sigma_1\ot \cdots \ot \sigma_{d-1}$ and $\forall i \ \widetilde{m}_i\equiv m_i-m$. The first term on the RHS of (\ref{6}) by the condition given by (\ref{typicality}) and the definition of the regularized relative entropy we get $\left(\frac{1}{d}+\delta \right)(1-\delta)E_r^{\infty}\left(\widetilde{\sigma}\right)$, which approaches $\frac{1}{d}E_r^{\infty}\left(\widetilde{\sigma}\right)$, since $\delta$ can be arbitrarily small. Now, we have to bound the rest of terms (the tail)  in (\ref{6}). To do so, we use the definition of the relative entropy of entanglement $E_r$ and the fact that $\mathbf{1}/d\in \mathcal{SEP}$:
	\be
	\label{Erbound}
	\forall \widetilde{\sigma} \quad E_r(\widetilde{\sigma}) \equiv\mathop{\inf}\limits_{\omega \in \mathcal{SEP}} D(\widetilde{\sigma}|\omega)\leq D\left(\widetilde{\sigma} | \frac{\mathbf{1}}{d}\right) ,
	\ee
	where $D(\cdot|\cdot)$ denotes the relative entropy. Thanks to this $D(\widetilde{\sigma} | \frac{\mathbf{1}}{d})=-\tr\widetilde{\sigma} \log \widetilde{\sigma} - \tr \widetilde{\sigma} \log (\mathbf{1}/d)=\log d-S(\widetilde{\sigma})\leq \log d$, since $\forall \widetilde{\sigma}  \ S(\widetilde{\sigma})\geq 0$, where $S(\cdot)$ is the von Neumann entropy. Because of that, and the subadditivity of $E_r$ we can bound every term from the tail in (\ref{6}) as $(1/\widetilde{m}_i) E_r(\widetilde{\sigma}^{\ot \widetilde{m}_i}) \leq E_r(\widetilde{\sigma})\leq \log d_i$. Finally, we are in the position to bound the tail in (\ref{6}), namely by (\ref{5}) and by above considerations the desired bound is of the form $2\delta(1-\delta)\sum_i \operatorname{log}d_i$. Now since  the sum $\sum_i \operatorname{log}d_i$ is constant and $\delta$ can be arbitrarily small we conclude that the tail in (\ref{6}) vanishes.
\end{proof}

\section{Proofs from Section~\ref{sec:properties}}
\begin{proof}[Proof of Observation~\ref{closed:sets}]
	For the proof it amounts to observe that $K_D$ is superadditive, $E_r$ is subadditive and $I_{sq}$ is additive \cite{Winter-squashed-ent},
	while the latter two are bounds on $K_D$ \cite{BB84,pptkey,keyhuge,Winter-squashed-ent} .
\end{proof}

\begin{proof}[Proof of Proposition~\ref{prop:charac}]
	We have by the definition $\mathcal{R}_{ir} \subseteq \mathcal{IR}$. We show that $\mathcal{IR} \subseteq \mathcal{R}_{ir}$ unless a bound key state exists.
	The proof goes by contradiction. Suppose there exists $\gamma_{0} \in \mathcal{IR}\setminus \mathcal{R}_{ir}$. We will show that it implies existence of
	bound key states. Let it have the form
	$ \gamma_{0} = {1\over d}\sum_{i,j=0}^{d-1} |ii\>\<jj|\ot U_i\sigma U_j^{\dagger}$ and $\sigma_j \equiv U_j\sigma U_j^{\dagger}$.
	For the state $ \gamma$ by Theorem \ref{th:main} we have
	\be
	K_D( \gamma_{0}) \geq \log d + {1\over d}K_D(\sigma_0\otimes \cdots \otimes \sigma_{d-1}),
	\ee
	since the protocol from Theorem \ref{th:main} may be suboptimal.
	
	Since $\gamma_{0}$ is from $\mathcal{IR}$, we have $K_D( \gamma_{0}) = \log d$
	which proves $K_D(\sigma_0\otimes \cdots \otimes\sigma_{d-1})=0$. By monotonicity of $K_D$ this proves, that for all $i$ we have
	$K_D(\sigma_i) = 0$, as $\sigma_i$ can be obtained from $\sigma_0\ot \cdots \ot \sigma_{d-1}$ by LOCC.
	
	On the other hand, $ \gamma_{0}$ does not belong to $\mathcal{R}_{ir}$, that is $K_D( \gamma_{0}) \neq E_r( \gamma_{0})$.
	More precisely $\log d < E_r( \gamma_{0})$, since $K_D\leq E_r$ in general. 
	Due to subadditivity of $E_r$ via Theorem \ref{thm:upper-bound-er}, there is
	\be
	\log d < E_r(\gamma_{0}) \leq \log d + {1\over d} E_r(\sigma_0\ot \cdots \ot \sigma_{d-1})
	\label{eq:old-er-bound}
	\ee
	and since $E_r( \gamma_{0}) > \log d$, we have by the above the inequality that
	\be
	E_r(\sigma_0\ot \cdots \ot \sigma_{d-1}) > 0.
	\ee
	From the faithfulness of the relative entropy of entanglement \cite{Piani} it follows that at least one $\sigma_i$ has to be entangled. Otherwise, if all $\sigma_i$ were separable, it would mean, that the product $\sigma_0\ot \cdots \ot \sigma_{d-1}$ is also separable and hence $E_r(\sigma_0\ot \cdots \ot \sigma_{d-1})=0$.
	Recall however, that as we argued above, $K_D(\sigma_{i_0}) = 0$ therefore $\sigma_i$ has the bound key.
	Thus if there are no bound key states, we have $\mathcal{IR} \subseteq \mathcal{R}_{ir}$ which proves $\mathcal{IR} = \mathcal{R}_{ir}$ under
	no key bound states condition as claimed.
\end{proof}

One can see that proof of Proposition~\ref{prop:charac} can be obtained using bound~\eqref{Er1} of Theorem~\ref{thm:upper-bound-er} for regularized relative entropy by substituting it into~\eqref{eq:old-er-bound} and using argument of faithfulness.

\begin{proof}[Proof of Corollary~\ref{cor4}]
	From paper \cite{HHH1997-distill} we know that any state in $2\otimes 2$ which is entangled, has non-zero distillable entanglement.  Since we work with irreducible private states as a direct consequence of  Proposition \ref{prop:charac} if follows that $K_D(\sigma_i)=0$ and therefore each state $\sigma_i$ is separable as there are no bound key states in $2\otimes 2$.
	The second part follows from the application of
	Lemma 3 of \cite{keyhuge}, as $\sqrt{XX^{\dagger}}$ and $\sqrt{X^{\dagger}X}$ are separable, hence have the positive partial transposition and this lemma applies.
\end{proof}

\section{Proofs from Section~\ref{sec:prop-kundist}}
\begin{proof}[Proof of Observation~\ref{obs-keyun}]
	$K_D(\gamma)=\log d$ implies via Theorem \ref{th:main}, that $K_D(\sigma_0\otimes \cdots \otimes\sigma_{d-1})\equiv K_D(\widetilde{\sigma})=0$.
	However the state $\hat{\gamma}$ can be obtained from $\widetilde{\sigma}$ via an LOCC operation. Indeed, attaching
	a state $\rho_a={1\over d}\sum_i |ii\>\<ii|$ is LOCC as the latter state is separable. Then making the control-partial trace over all but
	the $i-th$ state $\sigma_i$  from $\widetilde{\sigma}$ is an another bi-local LOCC operation. The partial trace over the ancillary state
	leaves us with $\hat{\gamma}$. Since the LOCC operation on the key-udistillalbe state leaves it key-undistillable, for $n=1$ the assertion follows. The proof
	for $n>1$ is analogous: just with $\rho_a^{\ot n}$ in the place of $\rho_a$.
\end{proof}

\begin{proof}[Proof of Theorem~\ref{thm:nec-ku}]
	Let us set $U_1\ot U_2$ arbitrarily. We have a chain of (in)equalities, which we explain below.
	\be
	\begin{split}
		&0 \geq C_{DW}((U_1\ot U_2\rho U_1^{\dagger}\ot U_2^{\dagger})_{ccq})\\
		& = C_{DW}( (U U_1\ot U_2\rho U_1^{\dagger}\ot U_2^{\dagger} U^{\dagger})_{ccq})\\ &\equiv C_{DW}((\rho')_{ccq})  \geq
		C_{DW}\left( \left( \tr_{A'B'}\rho'\right) _{ccq}\right).
	\end{split}
	\ee
	The first inequality comes from the fact that composition of $U_1\ot U_2$ and the Devetak-Winter protocol
	applied to subsystem of the state $U_1\ot U_2^{\dagger}\rho U_1\ot U_2^{\dagger}$ is a particular key-distillation protocol, which on a key-undistillable state cannot
	achieve a positive rate (in fact $C_{DW}$ as one can see in Eq. (\ref{eq:cdw}) can be negative). The ccq state considered there
	is with respect to subsystem $AB$ and computational basis, so as in the next equality.
	This equality is due to the aforementioned fact that unitary transformation $U^{\dagger}$ does not change the ccq state (w.r.t. to subsystem $AB$ and control-basis of $U$) of the state $U_1\ot U_2^{\dagger}\rho U_1\ot U_2^{\dagger}$ (see Theorem 3.3 \cite{KH-phd}). Now the partial trace of the subsystem $A'B'$ cannot increase the key, so that the final inequality holds, and the assertion follows.
\end{proof}

\begin{proof}[Proof of Corollary~\ref{cor:num4}]
	The statement of  corollary follows from the fact that the set of key-undistillable states is closed under LOCC operations.
	We also use the fact that $\eta(x)\leq h(x)$ with $h$ being the binary Shannon entropy function and that binary Shannon entropy is less than 1. 
	Idea of the embedding is based on least common multiple for dimensions. Suppose that $\Lambda(\rho)=\rho_{d_1}$ with $\dim \rho_{d_1}=d_1$. Then we define $\rho'\equiv \rho_{d_1}\ot |0\>\<0|^{\ot 2}_d$ and $\gamma'\equiv \gamma_d \ot |0\>\<0|_{d_1}^{\ot 2}$ with $\dim |0\>\<0|_d=d$, $\dim |0\>\<0|_{d_1}=d_1$. Because of that we have $\dim \rho'=\dim \gamma'$ and the trace norm $||\rho'-\gamma'||_1$ can be calculated. Moreover such embedding belongs to the set of $LOCC$ operations, since we add only additionally ancilla systems. 
\end{proof}
Now we are in the position to prove theorem \ref{thm:cbound}.
\begin{proof}[Proof of Theorem~\ref{thm:cbound}]
	We will use the fact, that $C_{DW}$ is asymptotic continuous. Clearly for the singlet state it reads $\log d$ and
	for the state emerging from the key-undistillable state via control-unitary and the partial trace, must be zero through Theorem \ref{thm:nec-ku}. This in turn implies the bound on the distance between states. Let us note here, that it is actually rather uncommon usage of the asymptotic continuity,
	aiming at the proof that states must be far as continuous faction differs on them enough, rather than finding evaluation of functions close on some close states.
	
	We first invoke the asymptotic continuity of the entropy and the conditional entropy function \cite{Alicki-Fannes,Synak05-asym}.
	As long as the states satisfy $||\rho_1 - \rho_2||_1 \leq {1\over 2}$ (each on the dimension $d$), we have:
	\be
	\left| S(\rho_1) - S(\rho_2)\right|  \leq ||\rho_1 -\rho_2||_1\log d + \eta\left( ||\rho_1 - \rho_2||_1\right) .
	\label{eq:s}
	\ee

	Let us fix $0<\epsilon' < 1$, and note, that $\gamma$ by the definition is of the form $U |\Psi_+\>\<\Psi_+| \ot \sigma_{A'B'} U^{\dagger}$ for $U=\sum_{ij} |ij\>\<ij|\ot U^{(ij)}_{A'B'}$, some state $\sigma$ on $A'B'$ and $|\Psi_+\>\<\Psi_+| =\sum_{ij}{1\over d}|ii\>\<jj|$ the singlet state. Let us note here, that we  assumed that the control basis of twisting $U$ is the computational basis, since the proof for other product basis $\{|e_i\>\ot|f_i\>\}$ is analogous.
	
	We first note that due to the monotonicity of operation under measurements, for an arbitrary state $\rho$ and the private state $\gamma$ there is:
	\be
	\epsilon' = ||\gamma - \rho ||_1 \geq \left| \left| \,|\Psi_+\>\<\Psi_+| - \tr_{A'B'} U^{\dagger}\rho U\right| \right| _1,
	\label{eq:monot}
	\ee
	where $\rho' := \tr_{A'B'} U^{\dagger}\rho U$ the result of "untwisting" $U^{\dagger}$ and the partial trace on $\rho$.
	
	As announced, we use now the fact that $C_{DW}\left( (|\Psi_+\>\<\Psi_+|)_{ccq}\right) = \log d$ and $C_{DW}\left( (\rho')_{ccq}\right)  \leq 0$ due to Theorem \ref{thm:nec-ku}, as $K_D(\rho)=0$.
	Hence, unless $\epsilon' > 1$, throughout asymptotic continuity of entropic functions,
	we have the following chain of inequalities:
	\be
	\begin{split}
		\log d \leq &\left| C_{DW}\left( (|\Psi_+\>\<\Psi_+|) _{ccq}\right) - C_{DW}\left( (\rho')_{ccq}\right) \right| \\ &\leq \left| \log d - S(\tilde{\rho}_A)\right|  + \left| \log d - S(\tilde{\rho}_B)\right| \\
		& +\left| \log d - S(\tilde{\rho}_{AB})\right|  + I(A:E)_{\tilde{\rho}},
	\end{split}
	\label{eq:cdw-bound}
	\ee
	where $\tilde{\rho} = (\rho')_{ccq}$ and we consider particular subsystems of $\tilde{\rho}$.
	Following \cite{KH-phd} we now observe, that $I(A:E)_{\rho'_{ccq}} \leq S(A:B)_{\rho'}$. Since $\rho'$ is close to the singlet state which is pure it satisfies:
	\be
	\begin{split}
		S(A:B)_{\rho'}\leq &||\,|\Psi_+\>\<\Psi_+| - \rho'||_1\log d^2\\
		& + \eta\left( ||\,|\Psi_+\>\<\Psi_+| - \rho'||_1\right).
	\end{split}
	\ee
	It is crucial to observe now, that the state $\rho'$ after measurement in basis of untwisting is equal to the $AB$  subsystem of the state $\widetilde{\rho}$, by definition of the action of taking a $ccq$ state. We have then:
	\be
	\begin{split}
		\epsilon_1&=||\widetilde{\rho}_A-\mathbf{1}/d||_1\leq \epsilon\\
		\epsilon_2&=||\widetilde{\rho}_B-\mathbf{1}/d||_1\leq \epsilon\\
		\epsilon_3&=\left| \left| \widetilde{\rho}_{AB}-\frac{1}{d}\sum_i |ii\>\<ii|\right| \right| _1\leq \epsilon
	\end{split}
	\ee
	and hence from~\eqref{eq:cdw-bound} and the above inequalities via~\eqref{eq:s} there is
	\be
	\begin{split}
		\log d\leq &\epsilon_1 \log d+\eta(\epsilon_1)+\epsilon_2 \log d+\eta(\epsilon_2)+\\
		&\epsilon_3 \log d^2+
		\eta(\epsilon_3)+\epsilon \log d^2+\eta(\epsilon).
	\end{split}
	\ee
	Now we use the fact, that function $\eta(x)$ is increasing for $x\in [0 \ , 0.367879]$. This assures that we can upper bound the RHS of above by $\epsilon$ in place of $\epsilon_1,\epsilon_2,\epsilon_3$. We then obtain:
	\be
	\begin{split}
		\log d &\leq 2\epsilon \left( \log d + \eta(\epsilon)\right)  + 2\left(\epsilon \log d^2 + \eta(\epsilon)\right)  \\
		&= 6\epsilon\log d + 4\eta(\epsilon)
		\label{eq:epsilons}
	\end{split}
	\ee
	which implies
	\be
	\epsilon \geq \frac{1}{6} - \frac{2 \eta(\epsilon)}{3\log d},
	\ee
	as claimed.
\end{proof}
\section{Proofs from Section~\ref{sec:approx}}
\begin{proof}[Proof of Observation~\ref{obs2a}]
	For the first part of theorem we base on results of section \ref{sec:prop-kundist}. Namely, from the proof of Theorem \ref{thm:nec-ku} it follows that
	\be
	K_D(\rho) \geq C_{DW}( (\tr_{A'B'}(U^{\dagger}\rho U))_{ccq})
	\ee
	where $U$ is the unitary which defines $\widetilde{\gamma}$.
	Now if $||\rho - \widetilde{\gamma} ||= \epsilon' \leq \epsilon$  by asymptotic continuity of $C_{DW}$ obtained in inequalities (\ref{eq:cdw-bound}-\ref{eq:epsilons}), there is
	\be
	C_{DW}((\tr_{A'B'}(U^{\dagger}\rho U)_{ccq}) \geq \log d - 6\epsilon' \log d - 4\eta(\epsilon'),
	\ee
	which for $\epsilon'\leq \epsilon \leq {1\over 2}$ implies the same inequality for $\epsilon$ due to monotonicity of the function $\eta$ on interval $[0 \ ,0.367879]$ This proves the first part of thesis of this theorem.
	For the second part, by the assumption there exists an irreducible private state $\tilde{\gamma}$ close by $\epsilon$ in the trace norm to $\rho$ which has the property that $M(\rho) = \log d$. We use the asymptotic continuity of $M$ to obtain that $\left| M(\rho) - M(\tilde{\gamma})\right| \leq O\left( \epsilon \log d+ h(\epsilon)\right) $. Hence
	$K_D(\rho) \leq M(\rho) \leq \log d + O\left( \epsilon \log d + h\left( \epsilon\right) \right) $.
\end{proof}
\section{Multipartite private states}
\label{sec:c}
\begin{definition} (compare \cite{AugusiakH2008-multi})
Let $U_{i}$ be some unitary operations for every $i$ and
let $\varrho_{A'B_1'\ldots B_l'}$ be a density matrix acting on $\mathcal{H}'$. By the {\it
multipartite private state} or {\it multipartite pdit} of $l+1$ subsystems we mean the
following
\begin{equation}\label{mpbit}
\gamma_{AB_1\ldots B_lA'B_1'\ldots B_l'}^{(d)}=\frac{1}{d}\sum_{i,j=0}^{d-1}|{{i}^{(1)}\ldots
{i}^{(l)}}\>\<{{j}^{(1)}\ldots {j}^{(l)}}|\ot
U_{i}\varrho_{A'B_1'\ldots B_l'}U_{j}^{\dag}.
\end{equation}
\end{definition}
Naturally, for $l=1$ the above reproduces the bipartite private
states. Moreover, in that case, the multipartite distillable key satisfies \cite{AugusiakH2008-multi}
\begin{equation}\label{DWbound1}
C_{D}(\varrho_{AB_1\ldots B_lE})\geq
\min_{j\in\{1,\ldots,l\}}I(A\!:\!B_{j})(\varrho_{AB_{j}E}^{(\mathrm{cqq})})-
I(A\!:\!E)(\varrho_{AB_{j}E}^{(\mathrm{cqq})}).
\end{equation}
Here, $\varrho_{AB_{j}E}^{(\mathrm{cqq})}$ denotes the cqq
state, which arises from the general one by tracing out all the
parties but the first and $j$th one and Eve.

\subsection{Proof of Theorem \ref{th:main} for multipartite private states }

In this section we will argue, that the same lower bound as given in Theorem \ref{th:main} holds in the case of multipartite private states.

We show this by studying the rate of the analogous protocol as described in Section \ref{sec:protocol},
with $B \rightarrow B_1,\ldots,B_l$ for some finite number $l$ of the parties, so that the total private state
has $l+1$ subsystems ($l$th+1 denoted as A for the Alices' subsystem). We note that the protocol is realizable,
since multipartite private state has maximally correlated subsystems in secure basis, so that each
of the parties can perform control-sorting of the conditional states $\sigma_i$ which is the main step of the protocol.

We will show now the cases when the proof for bipartite case needs to be modified. The state in (\ref{eq:state-in}) is the same with the only change that 
\begin{enumerate}
\item  $\rho_t = \sum_i \frac{1}{|Q_t|} |i{\bf i}\>\<i{\bf i}|_{AB_1...B_l}$,
where $\bf i$ is the multiindex with $l$ values $i$ each for one of several Bob's system 
\item the error state has more error flags:
$|e\>_A\ot|e\>_{B_1}\ot \cdots \ot|e\>_{B_l}\ot|e\>_E$. 
\end{enumerate}

These two changes are the only modifications of the state $\rho_{in}$ of Eq. (\ref{eq:state-rhoin}) to the multipartite case. Further, although the state $E'$ is different than in bipartite case, due to larger number of communication, we never used actual form of its state, hence it is still denoted as $\rho_{E'}$. In particular in Eq. (\ref{eq:close-alf}) the only difference is that the state $\alpha_{ideal}$ is of the form:
\be
\alpha_{ideal} = {1\over d_{t(m)}}\sum_{k=0}^{d_{t(m)}-1} |k {\bf k}\>\<k{\bf k}|_{AB_1...B_l} \ot \rho_{E'}
\ee
with ${\bf k} = k\ldots k$ for $l$ systems.
In this step we have used analogous fact to $K_D = C_D$ in multipartite case give in \cite{AugusiakH2008-multi}. Now the analogues of states $\rho_m$ and $\rho_m'$
(denote them $\rho_m{(l)}$ and ${\rho_m{(l)}}'$), differ from $\rho_m$ and $\rho_m'$ respectively only by the two aforementioned changes $1$ and $2$ as well
as difference in definition of $\alpha_{ideal}$ and $\hat{\alpha}_{ideal}$.
It is then clear that
\be
||\rho_{m}{(l)} - {\rho_m{(l)}}'||_1\leq \epsilon_{t(m)}.
\label{eq:mmprimclose}
\ee
Further, we have by analogous considerations and the inequality \ref{DWbound1} the lower bound:
\be
K_D(\gamma_d) \geq \frac{1}{m}\left (\min_{i \in \{1,...,l\}} I(AA':B_iB_i')_{{\rho_m{(l)}}}  - I(AA':E'')_{{\rho_m{(l)}}}\right) \equiv \frac{1}{m}\tilde{R}_{DW}({\rho_m{(l)}}).
\ee
Without loss of generality we may assume that the above minimization is realized for $i=1$. It is now crucial to observe, that
\be
\forall_{i \in \{1,...,l\}}\,\, {\rho_{m|AA'B_iB_i'E''}{(l)}}'= {\rho_{m|AA'B_1B_1'E''}{(l)}}'
\ee
where ${\rho_{m|AA'B_iB_i'E''}{(l)}}'$ is the state ${\rho_{m}{(l)}}'$ after tracing out all systems $B_j \neq B_i$. We then use the fact that due to
(\ref{eq:mmprimclose}), the state $\rho_{m|AA'B_1B_1'E''}(l)$ is close in trace norm to ${\rho_{m|AA'B_1B_1'E''}(l)}'$, so that we can compute
$\tilde{R}_{DW}(\rho_{m}(l)')$ and use asymptotic continuity of the mutual information to conclude about $\tilde{R}_{DW}(\rho_{m}(l))$.
It is then crucial to observe that $\rho_{m|AA'B_1B_1'E''}(l)'$ has the same form  as $\rho_m'$. Thus, we conclude that the main claim of the Lemma \ref{lemma-cdw} stays the same in multipartite case. What changes is the function $f(\epsilon_m, d_{t(m)}, d, m)$, with a substitution $d' \rightarrow d'^l$, as there are $l$ parties. However we still have $\frac{1}{m}f(\epsilon_m, d_{t(m)}, d, m) \rightarrow 0$ for the same $\delta(m)$ as in the bipartite case. It is just more demanding for $\epsilon_m$ and $\delta(m)$ to make $l \log d'$ zero, which is yet possible as long as $l$ is constant. We have then proved, that the proof for the multipartite case is {\it mutatis mutandis} equivalent to the bipartite one.
\end{document}